\documentclass[12pt,draftclsnofoot,onecolumn]{IEEEtran}

\usepackage{cite} 
\usepackage{xspace}
\usepackage{filecontents}

\usepackage{float}
\usepackage{pgf, tikz}

\usepackage{graphicx}
\usepackage{soul}

\usepackage[cmex10]{amsmath}

\usepackage{bm}

\usepackage{amssymb} 

\usepackage{array}

\usepackage{subfigure}
\usepackage{caption2} 

\usepackage{amsfonts}
\usepackage{amsthm}   

\usepackage{enumerate} 

\usepackage{extarrows}

\newtheorem{proposition}{Proposition}

\usetikzlibrary{patterns,snakes}

\usepackage{enumerate}

\usepackage{extarrows}

\newcommand{\myexpect}[2]{\mathbb{E}_{#1}\left[ #2 \right]}

\newcommand{\myprobability}[1]{\mathrm{Pr}\left\lbrace #1 \right\rbrace}
\newcommand{\myindex}[1]{\mathbf{1}\left\lbrace #1 \right\rbrace}

\newcommand{\myP}{\mathcal{P}}

\newcommand{\Pbpsk}{P_{\text{BPSK}}}

\newcommand{\Pout}{P_{\text{out}}}

\newcommand{\Poutasyn}{P^\text{asyn}_{\text{out}}}

\newcommand{\Echip}{\mathcal{E}_{\text{chip}}}

\newcommand{\Pdis}[1]{p^{\mathrm{d}}_{#1}}
\newcommand{\Pcon}[1]{p^{\mathrm{c}}_{#1}}

\newcommand{\Prx}[2]{\mathcal{P}^{(#1)}_{\mathrm{rx}}{#2}}
\newcommand{\Prxnew}[2]{\breve{\mathcal{P}}^{(#1)}_{\mathrm{rx}}{#2}}
\newcommand{\Erx}[2]{\mathcal{E}^{(#1)}_{\mathrm{rx}}{#2}}
\newcommand{\Erxnew}[2]{\breve{\mathcal{E}}^{(#1)}_{\mathrm{rx}}{#2}}
\newcommand{\EEH}[2]{\mathcal{E}^{(#1)}_{\mathrm{eh}}{#2}}

\newcommand{\Prxa}{\mathcal{P}_{0}}
\newcommand{\Prxb}{\mathcal{P}_{1}}
\newcommand{\Prxc}{\mathcal{P}_{2}}
\newcommand{\Prxd}{\mathcal{P}_{3}}

\newcommand{\Ptag}{P_{\text{tag}}}
\newcommand{\Ptagasyn}{P^{\text{asyn}}_{\text{tag}}}

\newcommand{\Ptagi}[1]{P^{({#1})}_{\text{tag}}}

\newcommand{\etabs}{\eta (1-\rho)}

\newcommand{\Peh}{\mathcal{P}_{\text{eh}}}
\newcommand{\PehK}[1]{\mathcal{P}_{\text{eh},{#1}}}

\newcommand{\Etag}{\mathcal{E}_{\text{tag}}}
\newcommand{\Etagasyn}{\mathcal{E}^{\text{asyn}}_{\text{tag}}}

\newcommand{\Preader}{P_{\text{reader}}}

\newcommand{\Preaderasyn}{P^{\text{asyn}}_{\text{reader}}}

\begin{document}

\title{Full-Duplex Backscatter Interference Networks Based on Time-Hopping Spread Spectrum}
\author{Wanchun Liu, Kaibin Huang, Xiangyun Zhou and Salman Durrani     \thanks{\setlength{\baselineskip}{13pt} \noindent Wanchun Liu, Xiangyun Zhou and Salman Durrani are with Research School of Engineering, The Australian National University, Canberra, ACT 2601, Australia 
	(emails: \{wanchun.liu, xiangyun.zhou, salman.durrani\}@anu.edu.au). 
Kaibin Huang is with the Dept. of Electrical and Electronic Engineering, The University of Hong Kong, Hong
Kong (email: huangkb@eee.hku.hk). }}
\maketitle

\vspace{-1.7cm}
\begin{abstract}
	\vspace{-0.3cm}
Future Internet-of-Things (IoT) is expected to wirelessly connect billions of low-complexity devices. 
For wireless \emph{information transfer} (IT) in IoT, high density of IoT devices and their ad hoc communication result in strong interference which acts as a bottleneck on wireless IT.
Furthermore, battery replacement for the massive number of IoT devices is difficult if not infeasible, making wireless \emph{energy transfer} (ET) desirable. 
This motivates:
(i) the design of \emph{full-duplex} wireless IT to reduce latency and enable efficient spectrum utilization, and
(ii) the implementation of \emph{passive} IoT devices using backscatter antennas that enable wireless ET from one device (reader) to another (tag).
However, the resultant increase in the density of simultaneous links exacerbates  the interference issue. 
This issue is addressed in this 
paper by proposing the design of full-duplex \emph{backscatter communication} (BackCom) networks, where a novel multiple-access scheme based on \emph{time-hopping spread-spectrum} (TH-SS) is designed to enable both one-way wireless ET and two-way wireless IT in coexisting backscatter reader-tag links. 
Comprehensive performance analysis of BackCom networks is presented in this paper, including forward/backward bit-error rates and wireless ET efficiency and outage probabilities, which accounts for energy harvesting at tags,  non-coherent and coherent detection at tags and readers, respectively, and the effects of asynchronous transmissions. 
\end{abstract}
	\vspace{-0.4cm}

\section{Introduction}
The vision of the future Internet-of-Things (IoT) is to connect tens of  billions of  low-complexity wireless devices (e.g., sensors and wearable computing devices), which are coordinated to enable new applications such as smart cities, home automation and e-healthcare. Designing the IoT physical layer faces several  challenges. For instance, battery replacement or recharging for the massive number of  IoT devices is difficult or even infeasible as many may be deployed in hazardous  environments or hidden in e.g., walls and furniture. Furthermore, ad hoc communications between dense IoT devices causes severe  interference that is a bottleneck on the network throughput. To tackle these challenges, the design of full-duplex \emph{backscatter communication} (BackCom) networks for supporting simultaneous one-way energy transfer (ET) and two-way information transfer (IT) in coexisting IoT links is proposed in this paper. The interference in such networks is suppressed by the proposed multiple-access scheme based on the \emph{time-hopping spread-spectrum} (TH-SS) that is also designed to facilitate wireless ET. Furthermore, the full-duplex IT is enabled by the superposition of coherent and non-coherent modulation given TH-SS and backscatter.

\subsection{Related Work}
BackCom has been mostly implemented for radio-frequency identification (RFID) applications where devices connected to the grid, called \emph{readers}, wirelessly power passive devices, called \emph{tags}, to feed back ID data \cite{OverviewRFID}. A tag harvests energy from an unmodulated carrier wave transmitted by the reader,  and modulates and reflects a fraction of the wave by adapting the level of  antenna impedance mismatch \cite{Boyer14}. This operation does not require the tag to have any active RF component (such as analog-to-digital converter, power amplifier and local oscillator) or an internal power source. This results  in  passive and extremely low-complexity hardware with small form factors, making BackCom a promising solution for implementing low-cost and ultra-dense IoT networks. 

\textit{BackCom Systems:} For this reason, active research has been conducted on designing techniques for various types of BackCom systems and networks which are more complex than the traditional~RFID systems\cite{Bletsas09,Bletsas14,Almaaitah14,Katabi12,Kaifeng16}. One focus of the research is to design multiple-access BackCom networks where~a single reader serves multiple tags. As proposed in \cite{Bletsas09}, collision can be avoided by directional beamforming at the reader and decoupling tags covered by the same beam using the frequency-shift keying modulation. Subsequently, alternative  multiple-access schemes  were  proposed in  \cite{Bletsas14} and \cite{Almaaitah14} based on  time-division multiple access and collision-detection-carrier-sensing based random access, respectively. A novel approach for collision avoidance was presented in  \cite{Katabi12} which treats backscatter transmissions by tags as a sparse code and decodes multi-tag data   using a compressive-sensing algorithm. 

\textit{ET in BackCom Systems:} IoT devices  having the capabilities of  sensing and computing consume more power than  simple RFID tags and also require much longer IT/ET ranges (RFID ranges are limited to only several meters). This calls for techniques for enhancing the ET efficiency in BackCom systems by leveraging the rich results from the popular area of  wireless power transfer (e.g., see the surveys in \cite{Huang15,Bi15}). In~\cite{Yang15}, it was proposed that  a reader is provisioned with multi-antennas to beam energy to multiple tags. An algorithm was also provided therein for the reader to estimate the forward-link channel, which is required for energy beamforming, using the backscattered pilot signal also transmitted by the reader. 

The wireless ET efficiency can be also enhanced by reader cooperation. For example, multiple readers are coordinated to perform ET (and IT) to multiple tags as proposed in \cite{Lee15}. The implementation of such designs require BackCom network architectures with centralized control. However, IoT relies primarily on distributed \emph{device-to-device} (D2D) communication. Large-scale distributed D2D BackCom are  modeled and analyzed in \cite{Kaifeng16} using stochastic geometry, where tags are wirelessly powered by dedicated stations (called power beacons). In particular, the network transmission capacity that measures the network spatial throughput was derived and maximized as a function of backscatter parameters  including duty cycle and reflection coefficient. Instead of relying on peer-to-peer ET, an alternative approach of powering IoT devices is to harvest ambient RF energy from transmissions by WiFi access points or TV towers~\cite{Ambient13}.

\textit{BackCom Systems with D2D Communication:} Conventional BackCom techniques designed for RFID applications mostly target  simple single-tag  systems  and one-way IT from a tag to a reader. 
{Nevertheless, 
for future IoT supporting D2D communications, many distributed reader-tag links will coexist. For example, near-by customers in a shopping mall/supermarket may use their smart devices (i.e., readers) to collect sales-promotion information from different goods on the shelf (which are sent by tags) at the same time.}
This prompts researcher to design  more sophisticated and versatile  BackCom techniques to improve the data rates and mitigate interference. An energy harvesting D2D BackCom system was designed in \cite{FullDuplex14} that features a full-duplex BackCom link where high-rate data and a low-rate control signal are transmitted in the  opposite directions using on-off keying and binary \emph{amplitude modulation} (AM), respectively, which are superimposed exploiting their asymmetric bit-rates. 
In BackCom systems with coexisting links, interference is a much more severe issue than that in conventional systems due to \emph{interference regeneration} by reflection at all nodes having backscatter antennas. 
{The TH-SS scheme was first proposed in an \emph{ultra-wide band} (UWB) system~\cite{MoeWin00},
and the idea of mitigating interference using TH-SS was further applied to UWB RFID (BackCom) systems~\cite{MoeWin10}.}
The drawback of such a system is that the required accurate analog detection of ultra-sharp UWB pulses places  a stringent requirement for hardware implementation and may not be suitable for low complexity IoT devices.

{Many application scenarios in the future IoT require low-latency transmissions, such as e-healthcare and public safety~\cite{IoTSurvey}. 
	Therefore, adopting full-duplex D2D communications between IoT devices such that each IoT node can speak and listen at the same time is desirable, since the latency of information transmission can be reduced significantly.
	The conventional approach for enabling full-duplex transmission over a single link relies on self-interference cancellation~\cite{FulldupexSurvey14}. The implementation requires sophisticated adaptive analog-and-digital signal processing that is unsuitable for low-complexity and low-power IoT devices. 
	Since in conventional BackCom systems, the reader/tag is able to transfer/receive energy and receive/transfer information simultaneously, it is natural to design a full-duplex BackCom system enabling simultaneous two-way information transmission/reception.
	A simple full-duplex BackCom design, supporting low data rates for IoT links, was proposed in\cite{FullDuplex14}.
	However, the drawback of the design is the requirement of asymmetric rates for transmissions in the opposite directions since it targets mixed transmissions of data and control signals.
	Though information flow in RFID applications is usually uni-directional, message exchange between nodes is common in IoT.
	Therefore, the reader-to-tag and tag-to-reader ITs are
	equally important for future IoT applications, which require symmetric communication links between the massive number of devices.}

\subsection{Contributions}
We consider a BackCom interference network comprising $K$ coexisting pairs reader-tag. Each reader is provisioned  with reliable power supply and performs both ET and IT to an intended tag that transmits data back to the reader by backscatter. Targeting this network, a novel multiple-access scheme, called \emph{time-hopping full-duplex BackCom},  is proposed to simultaneously  mitigate interference and enable full-duplex communication.  These two features are realized by two components of the scheme. 

First,  the  interference-mitigation feature of the  scheme relies on an extension of TH-SS to also support ET from readers to tags. 
{In general, there are two schemes to handle the interference: one is interference cancellation, and the other is interference suppression by spread spectrum.
The first scheme is not suitable for IoT since the reader and the tag are meant to be simple devices and are not able to apply complex analog circuits and digital algorithms to cancel the interference. Hence, it is better to adopt the second scheme and make the interference unlikely to happen.}
To this end, we propose the novel \emph{sequence-switch modulation}~scheme where a bit   is transmitted from a reader to a tag by switching between a pair of TH-SS sequences each containing a single random nonzero \emph{on-chip}. Besides reducing the interference power by the TH-SS sequence~\cite{Verdu99}, the design not only supports ET for every symbol via the transmission of a nonzero chip but also satisfies the constraint of non-coherent detection at tags using energy detectors~\cite{EnergyDetection67}. Each tag also continuously harvests energy from interference. 
{Although the TH-SS scheme has been widely adopted in UWB systems, the proposed sequence-switch modulation scheme has three novelties.
First, the sequence-switch modulation scheme facilitates the energy transfer, which is not considered in the UWB systems. Second, the joint rate and energy analysis is not considered in the studies of UWB system. Third, integrating the backscatter characteristics in the design is not considered in the studies of UWB system.}

Next, to realize the full-duplex feature of time-hopping BackCom, the backward transmission from a tag to a reader is implemented such that each tag modulates the transmitted on-chip in the corresponding TH-SS sequence using  the binary-phase-shift keying (BPSK)  and a reader performs coherent demodulation to detect the bit thus transmitted. The BPSK modulation at a tag is operated  by switching two impedances chosen according to the reflection coefficients having zero and $180$-degree phase shifts. Compared with the previous design of full-duplex BackCom in \cite{FullDuplex14}, the proposed technique has the advantages of supporting symmetric full-duplex data rates and interference mitigation. 

The performance of the proposed time-hopping full-duplex BackCom scheme is thoroughly analyzed in this paper in terms of \emph{bit-error rate} (BER) for IT and the expected \emph{energy-transfer rate} (ETR) and the \emph{energy-outage probability} for ET. The main results are summarized as follows:
\begin{enumerate}[1)]
\item (Synchronous Transmissions) First, consider a typical link in  a  two-link BackCom interference system where the time-hopping full-duplex BackCom scheme is deployed. Assume chip  synchronization between links.  From the IT perspective, the BERs for the forward (reader-to-tag) and the backward (tag-to-reader) transmissions are derived for both the cases of static and fading channels. The results quantify  the effects of TH-SS on mitigating the original and regenerated interference. Specifically, the BERs for forward transmission with non-coherent detection are shown to diminish  inversely with the \emph{sequence length} $N$  in the interference-limited regime. Moreover, the BERs for backward transmission with coherent detection converge to those   of classic BPSK, denoted as $P_{\text{BPSK}}$, approximately as $(P_{\text{BPSK}} + 1/N)$. From the ET perspective, the expected ETR and the energy-outage probability are derived. 
In particular, as $N$ increases, the expected ETR is observed to diminish if the power of nonzero chips are fixed, i.e., the power constrained case, or converge to the derived constants if the energy of each nonzero chip if fixed, i.e., the energy constrained case. 

\item (Asynchronous Transmissions) Next, the assumption of chip synchronization is relaxed. The preceding results are extended to the case of (chip) asynchronous transmissions. It is found that the lack of synchronization between coexisting links degrades the BER performance for both forward and backward transmissions. For example, in the high SNR regime, the BER for backward transmission is approximately doubled. Nevertheless, the effects of asynchronous transmissions on ET are  negligible.

\item ($K$-Link Systems) Last, the performance analysis for the two-link systems is generalized to a $K$-link system. It is shown that the BER for forward transmission as well as the expected ETRs are approximately proportional to $(K-1)$. 

\end{enumerate}
%
%

\textbf{Notation}: 
$\vert x \vert$ and $x^*$ denote the modulus and the conjugate of a complex number $x$, respectively.
$\mathfrak{R}\{\cdot\}$ denotes the real part of a complex number.
$\myexpect{}{X}$ denotes the expectation of a random variable $X$.
$\myprobability{\mathcal{A}}$ denotes the probability of the event $\mathcal{A}$. 
$\myindex{\mathcal{A}}$ denotes the indicator function, i.e., $\myindex{\mathcal{A}}$ is equal to one if $\mathcal{A}$ is true or otherwise, is equal to zero.

\section{System Model}\label{Sec_sys}
We consider a BackCom  system  consisting of $K$ coexisting single-antenna reader-tag pairs. Each reader is provisioned with a full-duplex antenna (see e.g.,\cite{Circulator13}) allowing simultaneous transmission and reception.  For simplicity, it is assumed that  self-interference (from transmission to reception) at the  reader  due to the use of a full-duplex antenna is perfectly cancelled,
{since the reader only transmits an unmodulated signal (i.e., the carrier wave), and the self-interference which can be easily cancelled by filtering in the analog domain.}
Each passive tag uses a backscatter antenna for transmission by backscattering a fraction of the incident signal and an energy harvester for harvesting the energy in the remaining fraction.  Each pair of intended reader and tag communicate by full-duplex transmission with robustness against interference using the design presented in the next section. The architecture of such a full-duplex passive tag is shown in Fig.~\ref{fig:circuit}.  The baseband additive white Gaussian noise (AWGN) at Reader~$k$ is represented by the random variable $z_{\text{reader},k}$ with variance $\sigma^2_{\text{reader}}$. The passband noise signal at Tag~$k$ is $z_{\text{tag},k}(t)$ with variance $\sigma^2_{\text{tag}}$.

It is assumed that all the readers/tags share the same band for communication. Block fading is assumed such that   the channel coefficients remain unchanged within a symbol duration  but may vary from symbol-to-symbol.
We consider both the static and Rayleigh fading channels, corresponding to the cases with or without mobility, respectively. 
{For the Rayleigh fading channel, we assume that the channel coefficient are composed of the large-scale path loss with exponent $\lambda$ and the statistically independent small scale Rayleigh fading.}	The distance between Reader~$m$ and Tag~$n$ is denoted by $d_{mn}$.
The \emph{channel state information} (CSI) of the intended backscatter channel (reader-to-tag-to-reader) is available at the corresponding reader. However, the CSI of interference channels is not available at the reader.
Moreover, tags have no knowledge of any channel.

It is important to note that the interference in a BackCom interference channel is more severe and complex than that in a conventional one. This is mainly due to \emph{interference regeneration} by backscatter antennas at tags that reflect all incident signals including both data and interference signals. As an example, a two-link system is shown in Fig.~\ref{fig:signals} where interference regeneration is illustrated.

The performance metrics are defined as follows. Both  the BER for the backward and forward IT are analyzed in the sequel. For the forward ET, we consider two metrics: (i) the \emph{expected} ETR, denoted as  $\Etag$ and  defined as the expected harvested energy at a tag per symbol,  and (ii)~the energy-outage probability, denoted as  $\Pout$ and defined as the probability that the harvested energy at the tag during a symbol duration is below the tag's fixed energy consumption, denoted as  $\mathcal{E}_0$. 

Note that $\Etag$ and $\Pout$ are related to  the cases of large or small energy storage at tags, respectively. 
{Specifically, a large energy storage battery accumulates the energy with random arrivals, and hence is able to constantly power the tag circuit especially when the instantaneous harvested energy is very small. 
Thus, we care about the expected harvested energy at a tag.
Given a small or no storage, the instantaneous harvested energy is required to exceed the circuit power so as to operate the tag circuit, and cannot be accumulated for further usage.
Thus, we care about the energy-outage probability.}
\begin{figure*}[t]
	\renewcommand{\captionfont}{\small} \renewcommand{\captionlabelfont}{\small}
			\hspace{-0.1cm}
	\minipage{0.55\textwidth}
	\centering
	\vspace{-1.5cm}
	\includegraphics[scale=0.8]{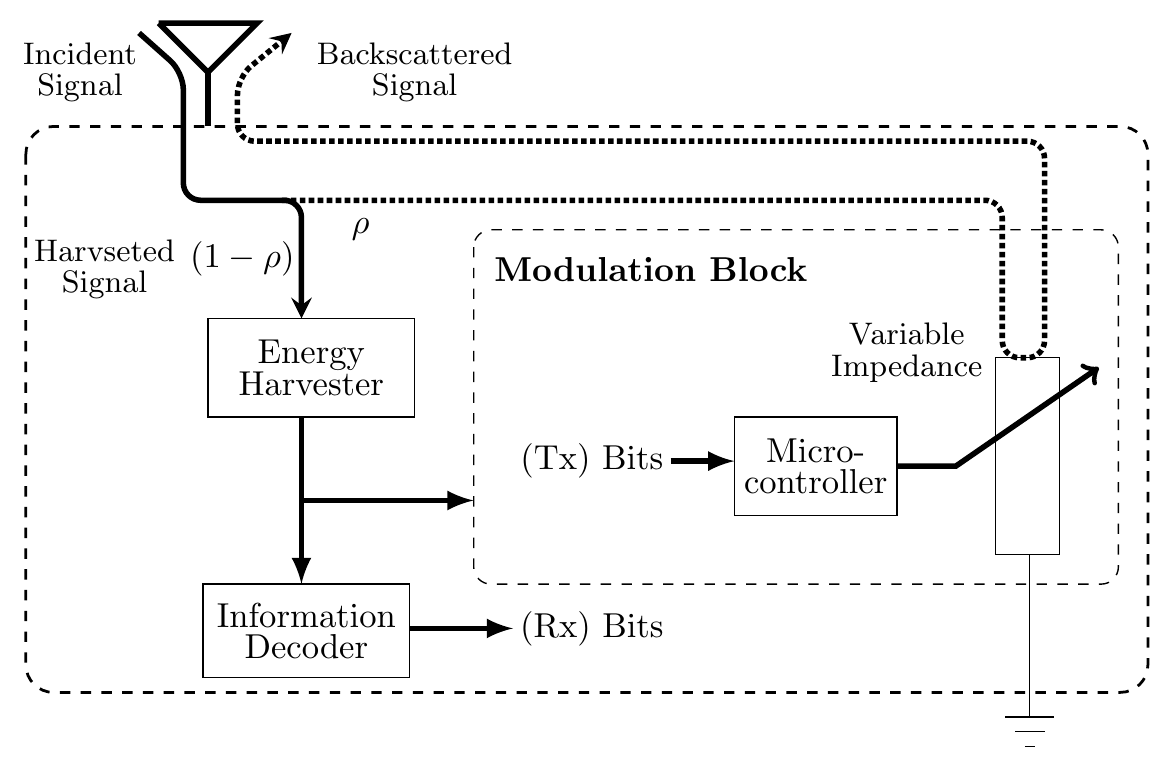}	
	\vspace*{-1.4cm}
	\caption{The architecture of a full-duplex tag.}
	\label{fig:circuit}
	\vspace*{-0.5cm}
	\endminipage
	\hspace{0.5cm}
	\vspace{-0.0cm}
	\minipage{0.43\textwidth}
	\centering
	\vspace{-1.0cm}	
	\includegraphics[scale=0.7]{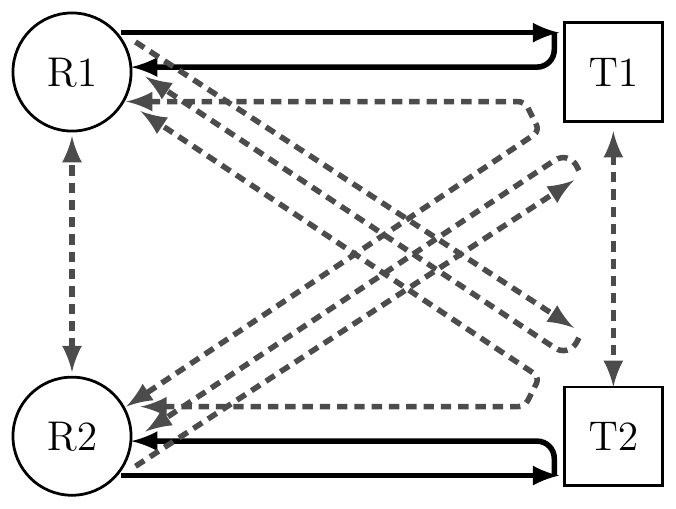}	
	\vspace*{-0.0cm}
	\caption{Two-link full-duplex interference channel. For example, signal R1-T1-R1 is the signal transmitted by Reader~1, then backscattered by Tag~1 and received by Reader~1.}\label{fig:signals}
	\endminipage
	\vspace*{-0.0cm}
\end{figure*}

\section{Time-Hopping Full-Duplex BackCom Scheme} \label{Sec_Design}
The proposed time-hopping full-duplex BackCom scheme comprises two components, namely the sequence-switch  modulation and full-duplex BackCom, which are designed in the following sub-sections. 

\subsection{Sequence-Switch Modulation}\label{Sec_Design_modulation}
The sequence-switch modulation used by each reader is designed for several purposes. The first is to suppress interference by TH-SS.   The second   is to enable simultaneous IT and ET under the constraint of non-coherent detection  at the intended tag by energy detection~\cite{EnergyDetection67}. Last, the modulation should support full-duplex BackCom by allowing a reader to transmit a carrier wave for backscatter by the intended tag.

Let a symbol duration $T$ be uniformly divided into $N$ slots called \emph{chips}.  For the mentioned  purposes, define a \emph{TH-SS sequence} as a $N$-chip random sequence comprising only a \emph{single} randomly located nonzero chip while others are  silent. Each link is assigned a pair of sequences with different nonzero chips to represent ``0" and ``1" of a bit. Then switching  the sequences enables the transmission of a binary-bit stream, giving the name of sequence-switch modulation.
{Note that the sequence-switch modulation is also named as the \emph{pulse-position modulation} in conventional UWB systems~\cite{MoeWin00}.}
Consider the generation of a pair of TH-SS sequences. The first sequence can be generated by randomly placing a (nonzero) \emph{on}-chip in one of the  $N$ chip positions and the second sequence by putting the corresponding  nonzero chip randomly in one of those chip-positions corresponding to zeros of the first sequence. A pair of TH-SS sequences for a particular link, say the $k$-th link, can be represented by the indices (or positions) of the corresponding pair of on-chips, denoted as $\mathcal{S}_k \triangleq \{s_{k0},s_{k1}\}$ and called a \emph{TH-SS pattern} [see Fig.~\ref{fig:chips}(a)], while all the other chips are the \emph{off-chips}. Note that there exist  $\frac{N(N-1)}{2}$ available patterns in total. The generation of the TH-SS patterns  for different links are assumed independent. The transmission of a single bit by Reader $k$ can be equivalently represented by a binary random variable $C_k$ with support $\mathcal{S}_k$, called a transmitted on-chip. Assuming chip synchronization between links, their transmissions in an arbitrary symbol duration can be represented by a set of i.i.d. random variables $\{C_k\}$ and illustrated in Fig.~\ref{fig:chips}(b). 

How the above design of sequence-switch modulation serves the mentioned purposes is discussed as follows. First,  interference between two links arises when their TH-SS patterns overlap and thereby causes detection errors at their intended tags. The likelihood of pattern overlapping reduces with the increasing {sequence length} $N$ (the processing gain) as the patterns become increasingly sparse {and different links are more likely to choose different patterns.} Consider two coexisting links with TH-SS patterns $\mathcal{S}_1$ and $\mathcal{S}_2$.  Given the design of sequence-switch modulation, there exist three scenarios for the relation between the two patterns, namely \emph{non-overlapping} ($\vert \mathcal{S}_1 \cap \mathcal{S}_2 \vert = 0$),
\emph{single-chip overlapping} ($\vert \mathcal{S}_1 \cap \mathcal{S}_2 \vert = 1$), and 
\emph{dual-chip overlapping} ($\vert \mathcal{S}_1 \cap \mathcal{S}_2 \vert = 2$). For each scenario, the actual transmitted on-chips may or may not \emph{collide} with each other. Thus each scenario can be further divided into multiple cases as illustrated in Fig.~\ref{fig:all_chips}.

 Next, the modulation design facilitates non-coherent detection at tags using energy detectors. 
 {For a particular link, since the assigned TH-SS pattern is known to both the reader and tag,  the tag detects the transmitted bit by estimating which of the two on-chips (i.e., the chips $s_{k0}$ and $s_{k1}$) in the pattern is transmitted using an energy detector. 
 Specifically, if the harvested energy in the chip $s_{k0}$ is larger than that in the chip $s_{k1}$, the estimated bit is `0', otherwise, it is `1'.}
 Furthermore,  the design of sequence-switch modulation enables ET simultaneous with IT by having an on-chip in every symbol duration for delivering energy to the intended tag. In addition, the tag also harvests energy from on-chips from the interference channels. As a result, the design achieves an ET efficiency at least twice of that by using the on-off keying, namely switching between a TH-SS sequence and a all-zero sequence. 
 {Note that unlike the conventional active full-duplex transceiver, the full-duplex tag's transmission and reception are passive backscattering based and energy detection based, respectively.
 	Thus, for such a passive transceiver, it is reasonable to assume that the backscattered (i.e., reflected) signal has no interference on the received signal for energy detection.}
 
 Last, this  design feature of having an on-chip in every data symbol facilitate full-duplex BackCom. Specifically, the carrier wave modulating each on-chip is modulated and backscattered by the intended tag for backward IT. The details are provided in the sequel.

\begin{figure*}[t]
	\renewcommand{\captionfont}{\small} \renewcommand{\captionlabelfont}{\small}
	\minipage{0.4\textwidth}
	\centering
	\vspace{-0.5cm}
	\includegraphics[scale=0.5]{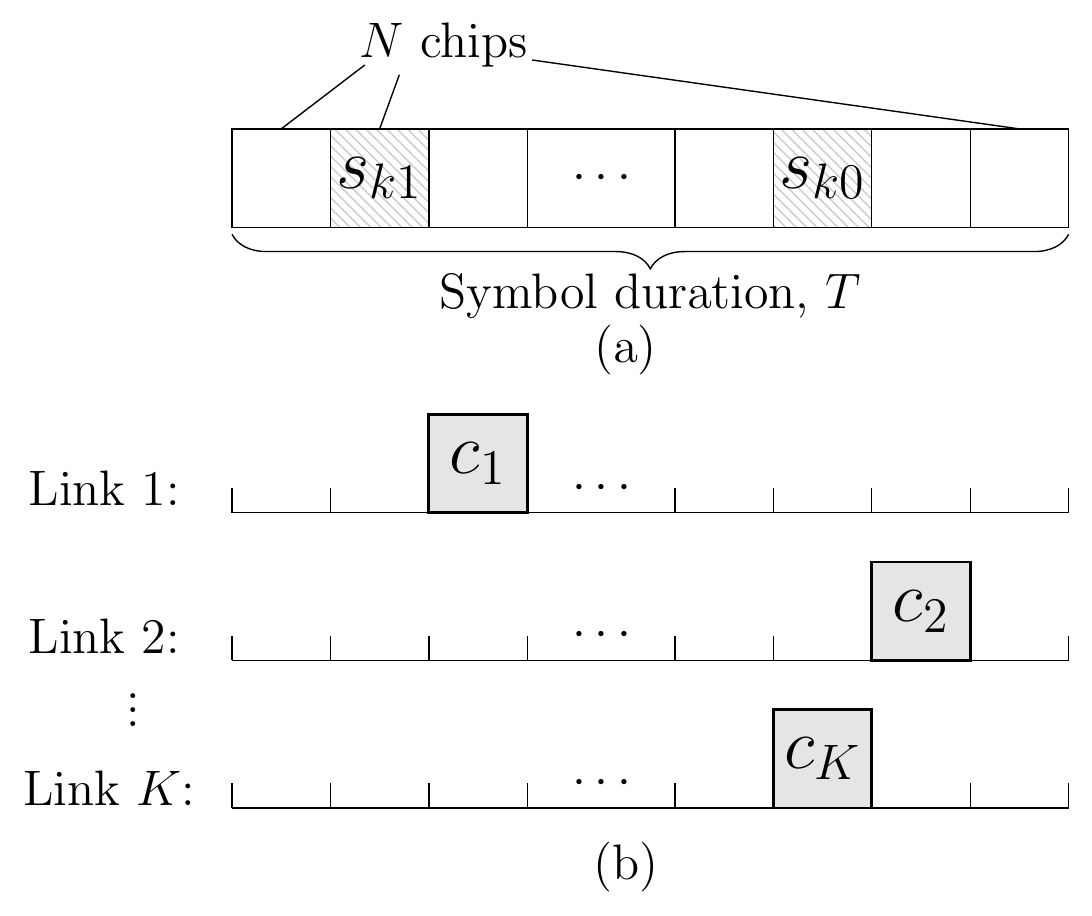}	
	\vspace*{-0.5cm}
	\caption{Sequence-switch modulation. (a) A TH-SS pattern. (b) Chip-synchronous transmissions of different links. }\label{fig:chips}
	\endminipage
	\hspace{0.6cm}
	\minipage{0.5\textwidth}
	\centering
	\vspace{-1.2cm}
	\includegraphics[scale=0.7]{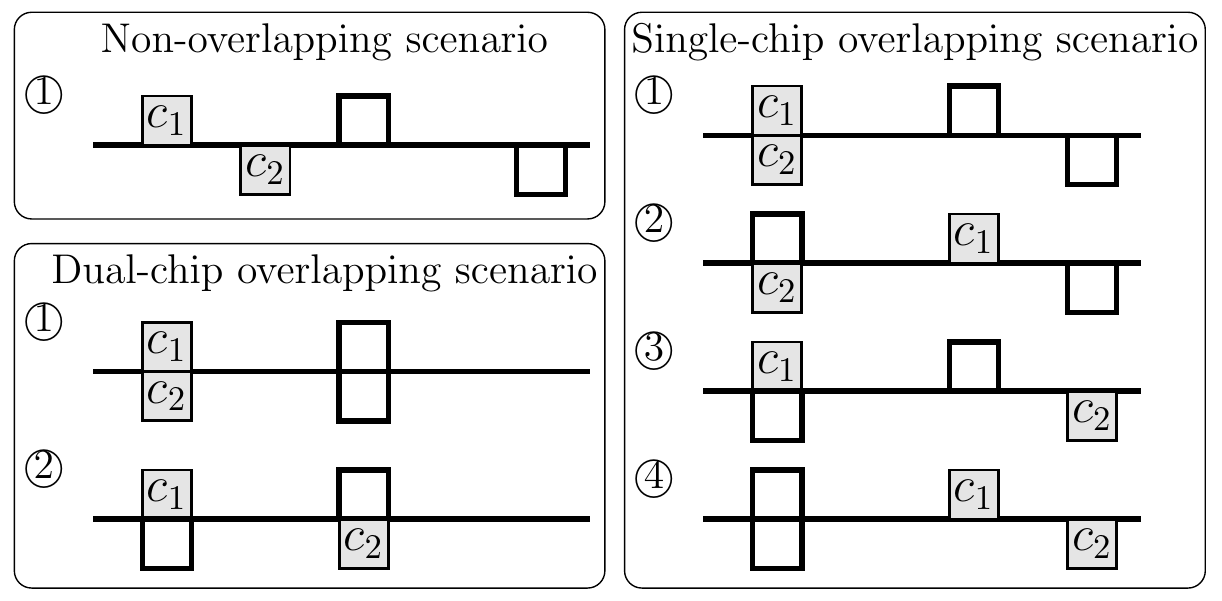}	
	\vspace*{-0.6cm}
	\caption{\small Two-link pattern-overlapping scenarios and transmission cases. The pairs of on-chips of Links~1 and~2 are the squares above and under the lines, respectively. }
	\label{fig:all_chips}
	\vspace*{-0.4cm}
	\endminipage
	\vspace*{-0.4cm}
\end{figure*}

\subsection{Full-Duplex BackCom} \label{Sec_Design_full_duplex}
Building on the sequence-switch modulation in the preceding subsection, the full-duplex BackCom is realized by the joint operation of the intended reader and tag designed as follows. 

Consider the reader side of a particular link. A reader using one full-duplex  antenna  to transmit a forward bit stream using the sequence-switch modulation by sending a carrier wave represented by $\sqrt{2 \myP} \mathfrak{R}\{ e^{j\omega t}\}$ during the on-chip, where  $\myP$ is the transmission  power and $\omega$ is the  angular frequency. At the same time, the reader receives the tag's backscattered signal at the same antenna. After cancelling the self-interference from its transmission, the reader detects the backward bit stream by BPSK demodulation/detection of the backscattered signals in the intervals corresponding to  the transmitted  on-chips that are known to the reader. The reader discards the received signals in other intervals since they are interference. Note that this operation requires the chip-level synchronization between a pair of intended reader and tag.

Next, consider the tag side of the link. During the off-chips of the assigned TH-SS pattern, the tag disconnects the modulation block and harvests energy from the other readers' transmissions (see the tag architecture in Fig.~\ref{fig:circuit}). Given the RF energy harvesting efficiency $\eta$, only $\eta$ portion of the RF receive power is harvested. During the two on-chip intervals of each instance of the pattern (or equivalently each symbol), the tag connects the modulation block and detects a forward bit by comparing the amounts of energy obtained from the two intervals using an energy detector.  Based on the detection results, the tag estimates  the transmitted on-chip positions and modulated/backscattered the signals in the corresponding chip intervals. 
Note that interference signal in either one or two of the on-chips may result in failure of the tag on detecting the  corresponding forward bit.
The variable impedance at the tag (see Fig.~\ref{fig:circuit}) is implemented by switching between two fixed impedances chosen to generate two reflection coefficients with the same magnitudes, namely $\sqrt{\rho}$, but different phase shifts, namely  zero and $180$ degrees. Then adapting the variable impedance to the backward bit streams modulates the backscatter signals with the bits by BPSK. 
Let the BPSK symbol transmitted by Tag $k$ be denoted as  $q_{k}$ with $q_{k}\in \{1, -1\}$. Thus, during the two on-chip intervals, the tag backscatters a fraction, denoted as $\rho$ with $\rho \in (0, 1)$, of the incident signal power and harvests the remaining fraction of $\eta (1-\rho)$. 

Combining the aforementioned reader and tag operations realize the full-duplex BackCom with symmetric backward and forward IT rates.

\section{Time-Hopping Full-Duplex BackCom: IT Performance}\label{Sec_IT}
In this section, we analyze the BER for the backward and forward IT. For simplicity,  a two-link BackCom system is considered as shown in  Fig.~\ref{fig:signals}. The results are  generalized for  the $K$-link sysetm  in Sec.~\ref{Sec_K}. The analysis in this section focuses on the typical link, Link~$1$, without loss of generality. 
\vspace{-0.5cm}
\subsection{BER at the Reader for Backward IT} \label{Sec_IT_backward}
Consider demodulation and detection of an arbitrary bit at Reader~1. 
As discussed in Sec.~\ref{Sec_sys} and Sec.~\ref{Sec_Design}, the receive baseband signal at Reader~1 during its transmitted on-chip $C_1$ can be written as
\begin{equation} \label{Reader_signals}
\begin{aligned}
r_1 
&= \sqrt{\myP} f_{11}  \sqrt{\rho}  b_{11} q_1 
+ \myindex{C_1 \in \mathcal{S}_2 } \sqrt{\myP}  f_{12} \sqrt{\rho} b_{21} q_2 \\
&+ \myindex{C_1 = C_2} \sqrt{\myP} \left(h_{21}  + f_{21}  \sqrt{\rho} b_{11}  q_1 + f_{22}  \sqrt{\rho} b_{21} q_2 \right) + z_{\text{reader},1},
\end{aligned}
\end{equation}
where the two indicator functions indicate whether Tag~2 is backscattering and whether Reader~2 is transmitting during Reader~1's transmitted on-chip $C_1$, respectively.
{$f_{mn}$ is the forward channel coefficient between Reader~$m$ and Tag~$n$.
$b_{mn}$ is the backward channel coefficient between Tag~$m$ and Reader~$n$. 
$h_{mn}$ is the channel coefficient between Readers~$m$ and $n$, while $g_{mn}$ is the channel coefficient between Tags~$m$ and~$n$. We also assume reciprocity between the forward and backward channels, i.e., $f_{mn} = b^*_{nm}$.}
The first term and the second~term~in~\eqref{Reader_signals} correspond to the useful signal R1-T1-R1 and the interference signal R1-T2-R1, respectively,~and~the third term corresponds to the interference signals R2-R1, R2-T1-R1, and R2-T2-R1 as~illustrated~in~Fig.~\ref{fig:signals}.

Given interference at Reader~1, the BER of coherent detection can be close to the maximum of $0.5$ and thus is assumed as $0.5$ when Reader~1 suffers interference for simplicity. Then the BER at Reader~$1$ can be written as 
\begin{equation} \label{slow_fading_reader_ber_0}
\Preader = 
\Pbpsk \myprobability{C_1 \notin \mathcal{S}_2} + 0.5\ \myprobability{C_1 \in \mathcal{S}_2}
=  \Pbpsk \left(p_0+ \frac{p_1}{2}\right) + \frac{1}{2}\left(\frac{p_1}{2} +p_2\right),
\end{equation}
where $\Pbpsk$ denotes the BER for BPSK detection without interference, and together with the probabilities $\{p_0, p_1, p_2\}$ are defined as follows: 
\begin{equation} \label{overlap_prob}
\begin{aligned}
&\Pbpsk  =  \myexpect{}{Q\left(\sqrt{{2 \myP \rho \vert  f_{11} b_{11} \vert^2}/{\sigma^2_{\text{reader}}} }\right)},\\
&p_0 \triangleq \myprobability{|\mathcal{S}_1 \cap \mathcal{S}_2| = 0} = {\dbinom{N-2}{2}}\Big/{ \dbinom{N}{2}} = \frac{(N-2)(N-3)}{N(N-1)},\\
&p_1 \triangleq \myprobability{|\mathcal{S}_1 \cap \mathcal{S}_2| = 1} = 1- {\dbinom{N-2}{2}}\Big/{ \dbinom{N}{2}} -  {1}\Big/{ \dbinom{N}{2}}= \frac{4(N-2)}{N(N-1)},\\
&p_2 \triangleq \myprobability{|\mathcal{S}_1 \cap \mathcal{S}_2| = 2} = {1}\Big/{ \dbinom{N}{2}}=\frac{2}{N(N-1)}.
\end{aligned}
\end{equation}
Substituting  \eqref{overlap_prob} into \eqref{slow_fading_reader_ber_0}, we obtain the BER for the backward IT at the reader: 
\begin{proposition} \label{proposition_1}
The expected BER for the backward IT is
\begin{equation} \label{slow_fading_reader_ber}
\Preader = \frac{N-2}{N} \Pbpsk +\frac{1}{N},
\end{equation}
where it can be derived straightforwardly based on the statistics of the channel coefficients that
\begin{equation} \label{PBPSK}
\Pbpsk
= \left\lbrace
\begin{aligned}
& Q\left(\sqrt{{2 \myP \rho \vert  f_{11} b_{11} \vert^2}/{\sigma^2_{\text{reader}}} }\right), &&\text{static channel,} \\
&\frac{1}{2} \left(1-\exp\left({\frac{d^{2\lambda}_{11} \sigma^2_{\text{reader}}}{4 \myP \rho }}\right) \mathrm{erfc}\left(\frac{d^{\lambda}_{11}}{2
} \sqrt{\frac{\sigma^2_{\text{reader}}}{\myP \rho} }\right)\right),  &&\text{Rayleigh fading channel,}
\end{aligned}
\right.
\end{equation}
and $Q(\cdot)$ and $\mathrm{erfc}(\cdot)$ are the Q-function and the complementary error function, respectively.
\end{proposition}

The BER for the backward IT decreases inversely with the reflection coefficient~$\rho$.
For the high SNR regime ($\myP/{\sigma^2_{\text{reader}}}\rightarrow \infty$), the BER reduces to $P_{\text{reader}}\approx \frac{1}{N}$, which is caused by the interference and decreases inversely with the TH-SS sequence length $N$. 

It is interesting to investigate the BER scaling law w.r.t. the sequence length $N$.
We consider two schemes when increasing $N$: (i) the fixed chip power (FCP) scheme which fixes the chip transmit power $\myP$, and (ii) the fixed chip energy (FCE) scheme which fixes the chip transmit energy denoted by $\Echip$, and $\Echip = \myP \frac{T}{N}$, i.e., $\myP$ increases linearly with $N$.

For the static channel, 
based on \eqref{slow_fading_reader_ber} and \eqref{PBPSK}, as the sequence length $N\rightarrow \infty$, the asymptotic BER for the FCP scheme is 
$ 
\Preader \approx \frac{N-2}{N} \Pbpsk,
$
and noise is the dominant factor for causing detection errors. 
While the asymptotic BER for the FCE scheme is 
$ 
\Preader \approx \frac{1}{N},
$
and interference is the dominant factor for causing detection errors.

For Rayleigh fading channel,
based on \eqref{slow_fading_reader_ber} and \eqref{PBPSK}, as the sequence length $N\rightarrow \infty$, the asymptotic BER expression for both the FCP and FCE schemes is
$ 
\Preader \approx \frac{N-2}{N} \Pbpsk,
$
and noise is the dominant factor for causing detection errors.

\vspace{-0.2cm}
\subsection{BER at the Tag for Forward IT}\label{Sec_IT_forward}
We investigate the BER at Tag~1 in three different TH-SS pattern overlapping scenarios.
Assuming that $\Ptagi{0}$, $\Ptagi{1}$ and $\Ptagi{2}$ are the BER conditioned on the events $\vert \mathcal{S}_1 \cap \mathcal{S}_2 \vert = 0$, $\vert \mathcal{S}_1 \cap \mathcal{S}_2 \vert = 1$ and $\vert \mathcal{S}_1 \cap \mathcal{S}_2 \vert = 2$, respectively, the BER at Tag~1 is
\begin{equation} \label{Tag_BER_first}
\Ptag = \sum\limits_{n=0}^{2} p_{n} \Ptagi{n}.
\end{equation}
We further calculate $\Ptagi{n}$ as follows (see pattern-overlapping scenarios in Fig.~\ref{fig:all_chips}):
\subsubsection{BER Given Non-Overlapping Scenario}
Tag~1's receive passband signal in the transmitted on-chip $C_1$ is
\vspace{-0.3cm}
\begin{equation}
y_1(t)= \sqrt{2\myP \etabs} \mathfrak{R} \{f_{11}  e^{j \omega t}\}  + z_{\mathrm{tag},1}(t).
\end{equation}
Thus, the receive signal power in the transmitted on-chip $C_1$ is
\begin{equation} \label{non_overlap_power}
\Prxa \triangleq \Prx{0}{} = \myP \etabs \vert f_{11} \vert ^2, 
\end{equation}
while since neither Reader~1 nor Reader~2 is transmitting during the other on-chip $\mathcal{S}_1\backslash C_1$, the receive signal power in the other on-chip is 
$
\Prxnew{0}{} = 0.
$

Based on \cite{EnergyDetection67}, scaling by the two-side power spectrum density of noise signal $z_{\text{tag},1}(t)$, the received energy during the transmitted on-chip $C_1$, $E_1$, follows a non-central chi-square distribution with $2$ degrees of freedom and parameter {$\gamma =\Prxa /{\sigma^2_{\text{tag}}}$, i.e., $\chi'^2(\gamma)$.}
Similarly, for the other on-chip $\mathcal{S}_1 \backslash C_1$, the scaled received energy $\breve{E}_1$, follows $\chi'^2(0)$.
Therefore, comparing the scaled receive energy between the chips $C_1$ and $\mathcal{S}_1 \backslash C_1$, i.e., $E_1$ and $\breve{E}_1$, the detection error probability is
\begin{equation} \label{non_overlap_ber}
\begin{aligned}
\Ptagi{0}
&= {1-\myprobability{\breve{E}_1 < E_1}} =
{1 - \int_{0}^{\infty} F_{\breve{E}_1} (x) f_{E_1}(x) \mathrm{d}x},
\end{aligned}
\end{equation}
where $F_{\breve{E}_1}(\cdot)$ and $f_{E_1}(\cdot)$ are the cumulative distribution function (cdf) and the probability density function (pdf) of the distributions $\chi'^2(0)$ and $\chi'^2( \Prxa/\sigma^2_{\text{tag}})$, respectively.
For generality, we define function $G(a,b)$ as
\begin{equation} \label{G_definition}
G(a,b) \triangleq 1- \myprobability{E_A < E_B}
= \int_{0}^{\infty} \frac{1}{2} Q_1(\sqrt{ a},\sqrt{x}) \exp\left(-(x+ b)/2\right) I_0\left(\sqrt{b x}\right) \mathrm{d}x,
\end{equation}
where $E_A$ and $E_B$ follows non-central chi-square distribution with $2$ degrees of freedom and parameters $a$ and $b$, respectively, and $Q_M(\cdot,\cdot)$ and $I_{\alpha}(\cdot)$ denotes the Marcum Q-function and the modified Bessel function of the first kind, respectively~\cite{Handbook}.
Thus, the BER in \eqref{non_overlap_ber} is represented~as 
\begin{equation} \label{non-overlap_ber}
\Ptagi{0} = {G(0, \Prxa/\sigma^2_{\text{tag}})}.
\end{equation}

\subsubsection{BER Given Single-Chip Overlapping Scenario}
There are four transmission cases each with the same probability (see transmission cases in Fig.~\ref{fig:all_chips}):

	\underline{Case 1}: The two readers are using the overlapping chip for transmission.
	Tag 1's receive signal in the transmitted on-chip $C_1$ consists of four signals, i.e., R1-T1, R1-T2-T1, R2-T1 and R2-T2-T1, thus, the receive signal power in the chip $C_1$ is  
	\begin{equation}
	\begin{aligned}
	&\Prxb  \triangleq \Prx{1}{(C_1=C_2=\mathcal{S}_1 \cap \mathcal{S}_2)} = \etabs \myP \left\vert f_{11}  + f_{12} \sqrt{\rho} g_{21} q_2 + f_{21}  + f_{22} \sqrt{\rho} g_{21} q_2  \right\vert^2,
	\end{aligned}
	\end{equation}
	while the receive signal power in the chip $\mathcal{S}_1 \backslash C_1$ is $\Prxnew{1}{(C_1=C_2=\mathcal{S}_1 \cap \mathcal{S}_2)} = 0$.
	Based on \eqref{G_definition}, considering the randomness of both the channel coefficients\footnote{Note that we have included all channel coefficients as the potential random~variables~over~which~the~expectation~is~taken. In Rayleigh fading channel, all channel coefficients are random variables, while in the static channel they are constants. For ease of presentation, we continue to use such notations in the rest of the paper.} and Tag~2's modulated signal, the BER~is
	\vspace{-0.5cm}
	\begin{equation} \label{single_overlap_1_ber}
	\Ptagi{1}(C_1=C_2=\mathcal{S}_1 \cap \mathcal{S}_2) =  \myexpect{f_{11},f_{12},f_{21},f_{22},g_{21},q_2}{G(0,\Prxb/\sigma^2_{\text{tag}})}.
	\end{equation}

	\underline{Case 2}: Reader~1 is using the non-overlapping chip, while Reader~2 is using the overlapping chip for transmission.
	The receive signal in the chip $C_1$ is signal R1-T1, since both Reader~2 and Tag~2 are not active in the chip, thus, the receive signal power in the chip $C_1$ is the same with~\eqref{non_overlap_power}, i.e.,
	$
	\Prx{1}{(C_1\neq C_2=\mathcal{S}_1 \cap \mathcal{S}_2)} = \Prxa.
	$
	While the receive signal in $\mathcal{S}_1 \backslash C_1$ consists of signals R2-T1 and R2-T2-T1, and
	\vspace{-0.2cm}
	\begin{equation} \label{single_overlap_2_power}
	\Prxc \triangleq \Prxnew{1}{(C_1\neq C_2=\mathcal{S}_1 \cap \mathcal{S}_2)} = \etabs \myP \left\vert f_{21}  + f_{22} \sqrt{\rho} g_{21} q_2 \right\vert^2.
	\end{equation}
	Thus, the BER~is 
	\vspace{-0.3cm}
	\begin{equation}
	\Ptagi{1}(C_1\neq C_2=\mathcal{S}_1 \cap \mathcal{S}_2) =  \myexpect{f_{21},f_{22},g_{21},q_2}{G(\Prxc/\sigma^2_{\text{tag}},\Prxa/\sigma^2_{\text{tag}})}.
	\end{equation}

	\underline{Case 3}: Reader~1 is using the overlapping chip, while Reader-2 is using non-overlapping chip for transmission.
	The receive signal in the chip $C_1$ consists of two signals, R1-T1 and R1-T2-T1, thus, the receive signal power in the chip is 
	\begin{equation} \label{single_overlap_3_power}
	\Prxd \triangleq \Prx{1}{(C_2\neq C_1=\mathcal{S}_1 \cap \mathcal{S}_2)} = \etabs \myP \left\vert f_{11}  + f_{12} \sqrt{\rho} g_{21} q_2 \right\vert^2,	
	\end{equation}
	and the receive signal power in the chip $\mathcal{S}_1 \backslash C_1$ is $\Prxnew{1}{(C_2\neq C_1=\mathcal{S}_1 \cap \mathcal{S}_2)} = 0$.
	Thus,~the~BER~is 
	\begin{equation} \label{single_overlap_3_ber}
	\Ptagi{1}(C_2\neq C_1=\mathcal{S}_1 \cap \mathcal{S}_2) =  \myexpect{f_{11},f_{12},g_{21},q_2}{G(0,\Prxd/\sigma^2_{\text{tag}})}.
	\end{equation}

	\underline{Case 4}: Reader~1 and Reader~2 are using non-overlapping chips for transmission. The receive signal power in the chips $C_1$ and $\mathcal{S}_1 \backslash C_1$ are $\Prxa$ and $0$, respectively, and thus, 
	\begin{equation}
		\Ptagi{1}(C_1 \neq \mathcal{S}_1 \cap \mathcal{S}_2,C_2 \neq \mathcal{S}_1 \cap \mathcal{S}_2) =  {G(0,\Prxa/\sigma^2_{\text{tag}})}.
	\end{equation}
\subsubsection{BER Given Dual-Chip Overlapping Scenario}
There are two transmission cases each with the same probability (see transmission cases in Fig.~\ref{fig:all_chips}):

	\underline{Case 1}: The two readers are using the same chip for transmission.
	We see that the receive signal power in the chips $C_1$ and $\mathcal{S}_1 \backslash C_1$ are
	$\Prxb$ and $0$, respectively, and thus,  
		\begin{equation}
			\Ptagi{2}{(C_1 = C_2)} = \myexpect{f_{11},f_{12},f_{21},f_{22},g_{21},q_2}{G(0,\Prxb/\sigma^2_{\text{tag}})}.
		\end{equation}

	\underline{Case 2}: The two readers are using different chips for transmission.
		The receive signal in the chip $C_1$ consists of signals R1-T1 and R1-T2-T1, thus, the receive signal power in the non-overlapping chip is the same with \eqref{single_overlap_3_power}, i.e.,
		$
		\Prx{2}{(C_1\neq C_2)} = \Prxd.
		$
		While the receive signal in $\mathcal{S}_1 \backslash C_1$ consists of signals R2-T1 and R2-T2-T1 which is the same with \eqref{single_overlap_2_power}, i.e., 
		$
		\Prxnew{2}{(C_1\neq C_2)} = \Prxc.
		$
		Thus, the BER is 
		\begin{equation}
		\Ptagi{2}(C_1\neq C_2) =  \myexpect{f_{11},f_{12},f_{21},f_{22},g_{21},q_2}{G(\Prxc/\sigma^2_{\text{tag}},\Prxd/\sigma^2_{\text{tag}})}.
		\end{equation}

\subsubsection{Main Results and Discussions}
Based on the analysis above and \eqref{Tag_BER_first}, the expected BER for the forward IT is 
\begin{equation} \label{first_ave_ber}
\begin{aligned}
\Ptag
&= \mathbb{E}_{f_{11},f_{12},f_{21},f_{22},g_{21},q_2}\left\lbrace \left(p_0 + \frac{1}{4}p_1\right) G(0, \Prxa /\sigma^2_{\text{tag}}) \right. \\
&+ p_1 \left(\frac{1}{4} G(0, \Prxb/\sigma^2_{\text{tag}}) + \frac{1}{4} G(\Prxc/\sigma^2_{\text{tag}}, \Prxa/\sigma^2_{\text{tag}}) + \frac{1}{4} G(0, \Prxd/\sigma^2_{\text{tag}})\right)\\
&\left.+ p_2 \left(\frac{1}{2} G(0, \Prxb/\sigma^2_{\text{tag}}) + \frac{1}{2} G(\Prxc/\sigma^2_{\text{tag}}, \Prxd/\sigma^2_{\text{tag}}) \right) \right\rbrace.
\end{aligned}
\end{equation}
For tractability, we consider the high SNR regime which means $\myP/\sigma^2_{\text{tag}} >> 0$, 
and ignore the noise effect on information detection, and thus,
$
G(a,b) \approx \myindex{a > b}.
$

From \eqref{first_ave_ber}, we further obtain
\begin{equation} \label{fading_tag_outage}
\begin{aligned}
&\Ptag
= \frac{N-2}{N(N-1)} \myprobability{\Prxc > \Prxa } + \frac{1}{N(N-1)} \myprobability{\Prxc > \Prxd }\\
&= \frac{N-2}{N(N-1)} \myprobability{\vert f_{21} + f_{22} \sqrt{\rho} g_{21} q_2\vert^2 >  \vert f_{11} \vert^2} \\
&+ \frac{1}{N(N-1)} \myprobability{\vert f_{21} + f_{22} \sqrt{\rho} g_{21} q_2 \vert^2 > \vert f_{11} + f_{12} \sqrt{\rho} g_{21} q_2 \vert^2}.
\end{aligned}
\end{equation}
Thus, as the sequence length $N \rightarrow \infty$,
\begin{equation} \label{high_N_Ptag}
\begin{aligned}
\Ptag
\approx 
\frac{1}{N} \myprobability{\vert f_{21} + f_{22} \sqrt{\rho} g_{21} q_2\vert^2 >  \vert f_{11} \vert^2 }.
\end{aligned}
\end{equation}

Therefore, increasing the sequence length $N$ reduces the BER for the forward IT. 

Since $q_2$ takes value with the same probability from $\{e^{j0}, e^{j\pi}\}$, we have the following result:
\begin{proposition} \label{first_ave_high_snr_ber}
For the static channel, the expected BER for the forward IT is 
\begin{equation} 
\begin{aligned}
&\Ptag
= \frac{N-2}{2N(N-1)} \left(
\myindex{ \vert f_{21} + f_{22} \sqrt{\rho} g_{21} \vert^2 >  \vert f_{11} \vert^2}
+
\myindex{ \vert f_{21} - f_{22} \sqrt{\rho} g_{21} \vert^2 >  \vert f_{11} \vert^2}
\right)\\
&\!+\!
\frac{1}{2N(\!N\!-\!1\!)} \!\left(\!
\myindex{\!\vert f_{21} \!+\! f_{22} \sqrt{\rho} g_{21} \vert^2 \!\!>\!\! \vert  f_{11}\! +\! f_{12} \sqrt{\rho} g_{21} \vert^2}
\!\!+\!
\myindex{\! \vert f_{21} \!-\! f_{22} \sqrt{\rho} g_{21} \vert^2 \!\!>\!\!  \vert f_{11} \!- \!f_{12} \sqrt{\rho} g_{21}  \vert^2}
\!\right)\!.
\end{aligned}
\end{equation}
\end{proposition}
Although the effect of reflection coefficient $\rho$ on the BER for the forward IT depends on the specific values of the channel coefficients, for the typical case that the channel between reader-tag pair is better than the cross reader-tag channel, i.e., $\vert f_{11} \vert^2 > \vert f_{21} \vert^2$, $\rho = 0$ minimizes $\Ptag$ to approach  zero since all the indicator functions in Proposition~\ref{first_ave_high_snr_ber} is equal to zero.

\begin{proposition} \label{Tag_BER_fading}
	For Rayleigh fading channel, the expected BER for the forward IT is 
\begin{equation} 
\begin{aligned}
\Ptag &=\! \frac{1}{N} 
\!-\!\frac{N\!-\!2}{N(N\!-\!1)}  \frac{1}{\rho} \left(\!\frac{d_{22}d_t}{d_{11}}\!\right)^\lambda
\!\exp\left(\!
\frac{d^\lambda_t}{\rho}
\left(\!
\left(\frac{d_{22}}{d_{21}}\right)^\lambda \!+\!
\left(\frac{d_{22}}{d_{11}}\right)^\lambda
\!\right)
\!\right)
\Gamma\left(\!
0,\ 
\frac{d^\lambda_t}{\rho}\left(\frac{d_{22}}{d_{21}}\right)^\lambda \!\!+\!
\left(\frac{d_{22}}{d_{11}}\right)^\lambda
\!\right)\\
&-\frac{1}{N(N\!-\!1)} \left(
\frac{d^\lambda_{22}}{d^\lambda_{12}+d^\lambda_{22}}
+\frac{d^\lambda_{t}}{\rho}  \frac{\frac{1}{d^\lambda_{11}d^\lambda_{22}}-\frac{1}{d^\lambda_{21}d^\lambda_{12}}}{\left(\frac{1}{d^\lambda_{12}}+\frac{1}{d^\lambda_{22}}\right)^2}
\exp\left(
\frac{d^\lambda_{t}}{\rho} \frac{\frac{1}{d^\lambda_{11}}+\frac{1}{d^\lambda_{21}}}{\frac{1}{d^\lambda_{12}}+\frac{1}{d^\lambda_{22}}}
\right)
\Gamma\left(
0,\ 
\frac{d^\lambda_{t}}{\rho} \frac{\frac{1}{d^\lambda_{11}}+\frac{1}{d^\lambda_{21}}}{\frac{1}{d^\lambda_{12}}+\frac{1}{d^\lambda_{22}}}
\right)
\right),
\end{aligned}
\end{equation}
where $d_t$ is the distance between Tag~1 and Tag~2, and $\Gamma\left(\cdot,\cdot\right)$ is the incomplete gamma function.
\end{proposition}

\begin{proof}
	See Appendix A.
\end{proof}
\noindent For the typical case that
$d_{11}<d_{21}$ and $d_{22}<d_{12}$, i.e., each reader-tag pair distance is smaller than the cross reader-tag distance,
we have $\frac{1}{d^\lambda_{11}d^\lambda_{22}}-\frac{1}{d^\lambda_{21}d^\lambda_{12}} > 0$ in Proposition~\ref{Tag_BER_fading}, and thus, it can be shown that $\Ptag$ monotonically increases with~$\rho$.

Therefore, for both the static and Rayleigh fading channels, a higher reflection coefficient leads to a higher BER for the forward IT in the typical case, and there is a clear tradeoff between the BER for the forward and backward transmission in terms of $\rho$.

\section{Time-Hopping Full-Duplex BackCom: ET Performance}
We analyze the expected ETR and the energy-outage probability in the following subsections.

\subsection{Expected ETR} \label{Sec_ET_ERT}
In Sec.~\ref{Sec_IT_forward}, we have analyzed the harvested power (energy) in the pair of on-chips. 
While for the off-chips, Tag~1 can harvest energy from Reader~2's IT and Tag~2's backscattering only if $C_2 \notin \mathcal{S}_1$. The probability $\myprobability{C_2 \notin \mathcal{S}_1} = p_0 + \frac{1}{2} p_1$.
Thus, Tag~1's receive signal in the chip $C_2$ consists of two signals when $C_2 \notin \mathcal{S}_1$, i.e., R2-T1 and R2-T2-T1, and the receive signal power in the chip is 
\begin{equation} \label{non-overlap_Eeh}
\Peh \triangleq \eta \myP \left\vert f_{21}  + f_{22} \sqrt{\rho} g_{21} q_2  \right\vert^2.
\end{equation}

Therefore, based on the analysis in Sec.~\ref{Sec_IT_forward}, considering Tag~1's receive signal power in the chips $C_1$, $\mathcal{S}_1 \backslash C_1$ and $C_2$, the expected ETR is 
\begin{equation} \label{first_ave_energy}
\begin{aligned}
\Etag
&= \frac{T}{N} \left(\sum_{n=0}^{2} p_{n} \myexpect{}{\Prx{n}{}+\Prxnew{n}{}} 
+ 
\myprobability{C_2 \notin \mathcal{S}_1} \myexpect{}{\Peh}
\right)\\
&=\frac{T}{N}
\myexpect{f_{11},f_{12},f_{21},f_{22},g_{21},q_2}{\frac{N-2}{N}\left(\Prxa + \Peh\right)+
	\frac{1}{N}\left(\Prxb+\Prxc+\Prxd\right)
}.
\end{aligned}
\end{equation}

\begin{proposition} \label{statiC_average_energy}
For the static channel, the expected ETR is 
\begin{equation} 
\begin{aligned}
&\Etag
=\frac{\eta \myP T }{N}
\left(
\frac{N-2}{N} \left((1-\rho) \vert f_{11} \vert^2 + \frac{1}{2}\left(\vert f_{21} +\sqrt{\rho} f_{22} g_{21} \vert^2 + \vert f_{21} -\sqrt{\rho} f_{22} g_{21} \vert^2\right)\right)  \right.\\
&\left.+
\frac{(1-\rho)}{2N} \left(\vert f_{11}+f_{21} + \sqrt{\rho} (f_{21}+f_{22}) g_{21} \vert^2 
+\vert f_{11}+f_{21} - \sqrt{\rho} (f_{21}+f_{22}) g_{21} \vert^2
\right) \right.\\
&\left.+\!
\frac{(1\!-\!\rho)}{2N}\left(\vert f_{21} \!+\!\sqrt{\rho} f_{22} g_{21} \vert^2 \!+\! \vert f_{21} \!-\!\sqrt{\rho} f_{22} g_{21} \vert^2
\!+\!
\vert f_{11} \!+\!\sqrt{\rho} f_{12} g_{21} \vert^2 \!+\! \vert f_{11} \!-\!\sqrt{\rho} f_{12} g_{21} \vert^2
\right) \right).
\end{aligned}
\end{equation}
\end{proposition}
Although the effect of reflection coefficient $\rho$ on the expected ETR for the forward ET depends on the specific values of the channel coefficients, for the typical case that direct channel signal is much stronger than the backscattered signal, i.e., $\vert f_{12} g_{21} \vert^2 << \vert f_{11}  \vert^2$ and $\vert f_{22} g_{21} \vert^2 << \vert f_{21}  \vert^2$, $\Etag$ increases inversely with~$\rho$.

\begin{proposition} \label{fading_average_energy}
	For Rayleigh fading channel, the expected ETR is 
\begin{equation} 
\begin{aligned}
\Etag
&= \frac{\eta \myP T}{N} \left(\nu_1 \rho^2 +\nu_2 \rho +\nu_3\right),\\
\end{aligned}
\end{equation}
where 
\begin{equation}
\hspace{-0.1cm}
\begin{aligned}
\nu_1 = - \frac{2}{N} \left(\frac{1}{d^\lambda_{12}d^\lambda_{t}} + \frac{1}{d^\lambda_{22}d^\lambda_{t}} \right),\ 
\nu_2 = \frac{2}{N}
\left( \frac{1}{d^\lambda_{12}d^\lambda_{t}}  - \frac{1}{d^\lambda_{21}}\right)
+
\frac{1}{d^\lambda_{22}d^\lambda_{t}}  - \frac{1}{d^\lambda_{11}},\ 
\nu_3 = \frac{1}{d^\lambda_{11}}+ \frac{1}{d^\lambda_{21}}.
\end{aligned}
\end{equation}
\end{proposition}
\begin{proof}
	See Appendix B.
\end{proof}
Thus, for the typical case that $d_{12}d_{t} >> d_{21}$ and $d_{22}d_{t} >> d_{11}$, one can show that $\Etag$ increases inversely with~$\rho$.

From Propositions~\ref{statiC_average_energy} and~\ref{fading_average_energy}, as the sequence length $N \rightarrow \infty$, the asymptotic expected ETR for the static channel and Rayleigh fading channel are given by
\begin{equation}
\begin{aligned}
	\Etag &\approx \frac{\eta \myP T }{N} \left((1-\rho) \vert f_{11} \vert^2 + \frac{1}{2} \left(\vert f_{21} +\sqrt{\rho} f_{22} g_{21} \vert^2 + \vert f_{21} -\sqrt{\rho} f_{22} g_{21} \vert^2\right) \right), \\
\Etag &\approx \frac{\eta \myP T}{N} \left( (1-\rho )\frac{1}{d^\lambda_{11}} +\frac{1}{d^\lambda_{21}} + \rho \frac{1}{d^\lambda_{22}d^\lambda_{t}}  \right),\\
\end{aligned}
\end{equation}		
respectively.

Therefore, for the FCP scheme, the expected ETR decreases with the sequence length $N$ and approaches zero,
while for the FCE scheme, the expected ETR converges to a constant with the increasing of sequence length.

\subsection{Energy-Outage Probability}
Based on the analysis in Sec.~\ref{Sec_IT_forward} and Sec.~\ref{Sec_ET_ERT}, the energy-outage probability at Tag~1 is 
\begin{equation} \label{first_Pout}
\begin{aligned}
&\Pout 
= 
p_0 \ \myprobability{ \frac{T}{N}  \left(\Prxa + \Peh\right) < \mathcal{E}_0 }
+ \frac{p_1}{4} \left( \myprobability{ \frac{T}{N}  \Prxb< \mathcal{E}_0} \right.\\
& \left. +  \myprobability{\frac{T}{N} \left(\Prxa + \Prxc\right)< \mathcal{E}_0} + \myprobability{ \frac{T}{N}  \left(\Prxd + \Peh\right)< \mathcal{E}_0} 
+  \myprobability{ \frac{T}{N}  \left(\Prxa + \Peh\right) < \mathcal{E}_0 }
\right)  \\
& + \frac{p_2 }{2}\left( \myprobability{\frac{T}{N}  \Prxb< \mathcal{E}_0} + \myprobability{\frac{T}{N}  \left(\Prxd + \Prxc\right)< \mathcal{E}_0}\right). 
\end{aligned}
\end{equation}

For the static channel, the energy-outage probability can be easily derived using \eqref{first_Pout} and is omitted here due to space limitations. 
The result for Rayleigh fading channel is presented in the following proposition.
\begin{proposition} \label{fading_outage}
	For Rayleigh fading channel, the energy-outage probability is
\begin{equation} 
\begin{aligned}
&\Pout
 = \frac{(N-2)^2}{N(N-1)} M\left(
\left(1-\rho \right)\frac{1}{d^\lambda_{11}},
0,
\frac{1}{d^\lambda_{21}},
\rho \frac{1}{d^\lambda_{22}} \frac{1}{d^\lambda_{t}} 
\right) +\! \frac{N\!-\!2}{N(N\!-\!1)} \times \\
&
\left(\!\!M\!\left(\!
(\!1\!-\!\rho)\frac{1}{d^\lambda_{11}},
0,
(\!1\!-\!\rho) \frac{1}{d^\lambda_{21}},
(\!1\!-\!\rho) \rho \frac{1}{d^\lambda_{22}}\frac{1}{d^\lambda_{t}}
\right)
\!+\! M\!\left(\!
(\!1\!-\!\rho)\frac{1}{d^\lambda_{11}},
(\!1\!-\!\rho) \rho \frac{1}{d^\lambda_{12}}\frac{1}{d^\lambda_{t}},
\frac{1}{d^\lambda_{21}},
\rho \frac{1}{d^\lambda_{22}}\frac{1}{d^\lambda_{t}}
\!\right)\!\!\right)\\
&+ \frac{1}{N(N-1)} M\left((1-\rho)\frac{1}{d^\lambda_{11}},
(1-\rho) \rho \frac{1}{d^\lambda_{12}}\frac{1}{d^\lambda_{t}},
(1-\rho) \frac{1}{d^\lambda_{21}},
(1-\rho) \rho \frac{1}{d^\lambda_{22}}\frac{1}{d^\lambda_{t}}
\right)\\
&+ \frac{1}{N} \tilde{M}\left(
(1-\rho)\left(\frac{1}{d^\lambda_{11}}+\frac{1}{d^\lambda_{21}}\right),
(1-\rho)
\rho \left(\frac{1}{d^\lambda_{12}}\frac{1}{d^\lambda_{t}} + \frac{1}{d^\lambda_{22}} \frac{1}{d^\lambda_{t}}\right)
\right),
\end{aligned}
\end{equation}
\end{proposition}
\noindent \emph{where }
\begin{equation} \label{outage_prob_2}
M(a,b,c,d)\!\triangleq 
\!1 \!-\! \int_{0}^{\infty} 
\frac{(a\!+\!bx) \exp\left(-\frac{\Xi}{(a\!+\!bx)} \!-\! x \right)\!-\!
	(c\!+\!dx) \exp\left(-\frac{\Xi}{(c\!+\!dx)} \!-\! x \right)
}
{a-c+(b-d)x}
\mathrm{d}x,\ \Xi= \frac{N \mathcal{E}_0}{\eta \myP T}
\end{equation}
\emph{and} 
\begin{equation} \label{M_new}
\tilde{M}\left(a,b\right) 
\triangleq 
1- \int_{0}^{\infty} \exp\left(-\frac{\Xi}{a+bx} - x\right) \mathrm{d} x.
\end{equation}

\begin{proof}
	See Appendix C.
\end{proof}
For the typical case that $d_{12}d_{t} >> d_{11}$ and $d_{22}d_{t} >> d_{21}$, i.e., 
each of the terms $M(\cdot)$ in \eqref{fading_outage} is approximated by $M\left(
\left(1-\rho \right)\frac{1}{d^\lambda_{11}},
0,
\frac{1}{d^\lambda_{21}},
0
\right)$ and the term $\tilde{M}(\cdot)$ is approximated by $\tilde{M}\left((1-\rho)\left(\frac{1}{d^\lambda_{11}}+\frac{1}{d^\lambda_{21}}\right),
0\right)$. Since both functions $M(\cdot,\cdot,\cdot,\cdot)$ and $\tilde{M}(\cdot,\cdot)$ decrease with each of the parameters, $\Pout$ increases with~$\rho$.

Based on \eqref{first_Pout}, as the sequence length $N \rightarrow \infty$, the asymptotic $\Pout$ is given by
\begin{equation}
\Pout \approx p_0 \ \myprobability{ \frac{T}{N}  \left(\Prxa + \Peh\right) < \mathcal{E}_0}
\approx
\myprobability{(1\!-\!\rho)\vert f_{11} \vert^2 \!+\! \vert f_{21} \!+\! \sqrt{\rho} f_{22}g_{21}q_2 \vert^2 < \Xi}.
\end{equation}
For the FCP scheme, as the sequence length $N \rightarrow \infty$, $\Xi \rightarrow \infty$ makes $\Pout \rightarrow 1$.
While for the FCE scheme, the asymptotic energy-outage probability for Rayleigh fading channel is given by
\begin{equation}
\begin{aligned}
\Pout & =  M\left(
\left(1-\rho \right)\frac{1}{d^\lambda_{11}},
0,
\frac{1}{d^\lambda_{21}},
\rho \frac{1}{d^\lambda_{22}} \frac{1}{d^\lambda_{t}} 
\right).
\end{aligned}
\end{equation}

\section{Performance of Time-Hopping Full-Duplex BackCom with Asynchronous Transmissions} \label{Sec_Asyn}
Considering the fact that the chip-synchronism is difficult to achieve in practical situations, 
in this section, we study BackCom with chip asynchronous transmissions. 
Without loss of generality, it assumes that $\tau$ is the delay shift between Links~1 and~2, which is positive and given by
\begin{equation}
\tau = \beta \frac{T}{N},\ \beta \in [0,1),
\end{equation}
where $\beta$ is named as the delay offset. Hence, the delay is assumed to be within a chip duration.

Due to the lack of perfect synchronization, the pattern-overlapping scenarios are more complex than that of the chip-synchronous case (Sec.~\ref{Sec_Design_modulation}). 
Considering that Link~1's TH-SS pattern $\mathcal{S}_1$ may consists of disjunct chips or consecutive chips illustrated in Figs.~\ref{fig:disjunct} and~\ref{fig:connected}, respectively.
Note that if the pattern $\mathcal{S}_1$ consist of the first chip and the last chip of a symbol, we say this pattern consists of consecutive chips.
Using $\Pdis{}$ and $\Pcon{}$ to denote the probability that Link~1's chips is (d)isjunct or (c)onsecutive, respectively, we have 
\begin{equation}
\Pdis{} = 1 -  \frac{2}{N-1} = \frac{N-3}{N-1},\  \Pcon{} = \frac{2}{N-1}.
\end{equation}
Then all the pattern overlapping scenarios are illustrated in Fig.~\ref{fig:asyn}, and $p^{a}_{n-i}$ denotes the probability of each sub-scenarios, where $a \in \{\mathrm{d},\mathrm{c}\}$, $n=0,1,2$ denotes the number of Link~2's chips overlapped by Link~1's chips, and $i=1,2,3,4$ denotes the pattern overlapping scenario index in Fig.~\ref{fig:asyn}. 
Note that the chip overlapping duration is $1-\beta$ of a chip in the single-chip overlapping scenario 1 of Fig.~\ref{fig:disjunct}, while it is $\beta$ in the single-chip overlapping scenario 2.
Thus, $p^{\mathrm{d}}_{n-i}$ and $p^{\mathrm{c}}_{n-i}$ can be obtained~as
\begin{equation} \label{prob_asyn}
\begin{aligned}
&\Pdis{0} = \Pdis{} \frac{(N-4)(N-5)}{N(N-1)},\ 
\Pdis{1-1} =\Pdis{1-2}= \Pdis{} \frac{4 N -16}{N(N-1)},
\Pdis{2-1} = \Pdis{2-2}= \Pdis{} \frac{2}{N(N-1)},\\ 
&\Pdis{2-3} =\Pdis{2-4} = \Pdis{} \frac{4}{N(N-1)},\
\Pcon{0} \!=\! \Pcon{} \frac{(N\!-\!3)(N\!-4)}{N(N\!-\!1)},\ 
\Pcon{1-1} \!=\! \Pcon{1-2}\!=\!\Pcon{1-3}\!=\! \Pcon{} \frac{2 N \!-\!6}{N(N\!-\!1)},\\
&\Pcon{2-1} \!=\!\Pcon{2-2}\!=\!\Pcon{2-3}\!=\! \Pcon{} \frac{2}{N(N-1)}.
\end{aligned}
\end{equation}

\begin{figure}[t]
	\small
	\renewcommand{\captionlabeldelim}{ }
	\renewcommand{\captionfont}{\small} \renewcommand{\captionlabelfont}{\small}		
	\centering     
	\vspace{-1cm}
	\subfigure[Link~1 uses disjunct chips.]{\label{fig:disjunct}
	\includegraphics[scale=0.4]{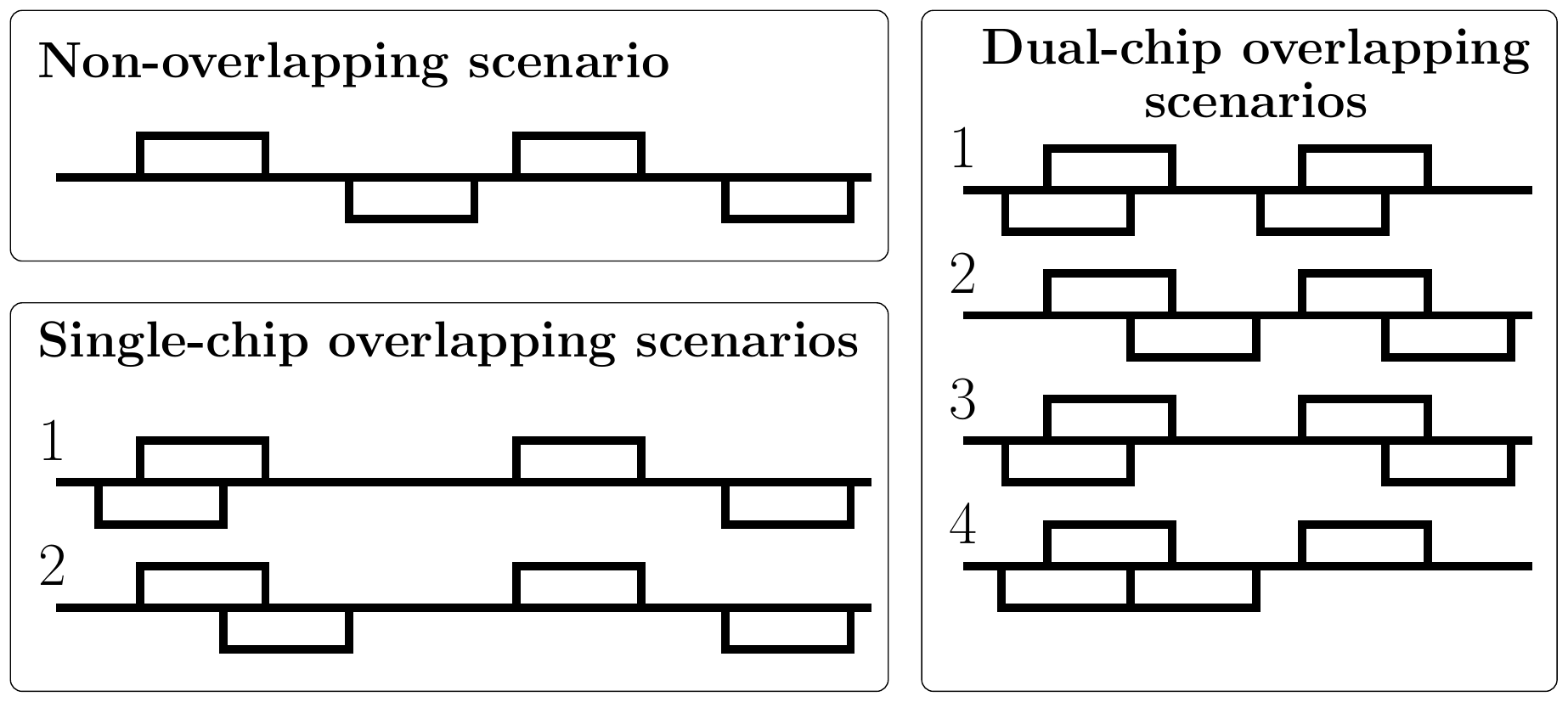}	
		}
	\subfigure[Link~1 uses consecutive chips.]{\label{fig:connected}
	\includegraphics[scale=0.4]{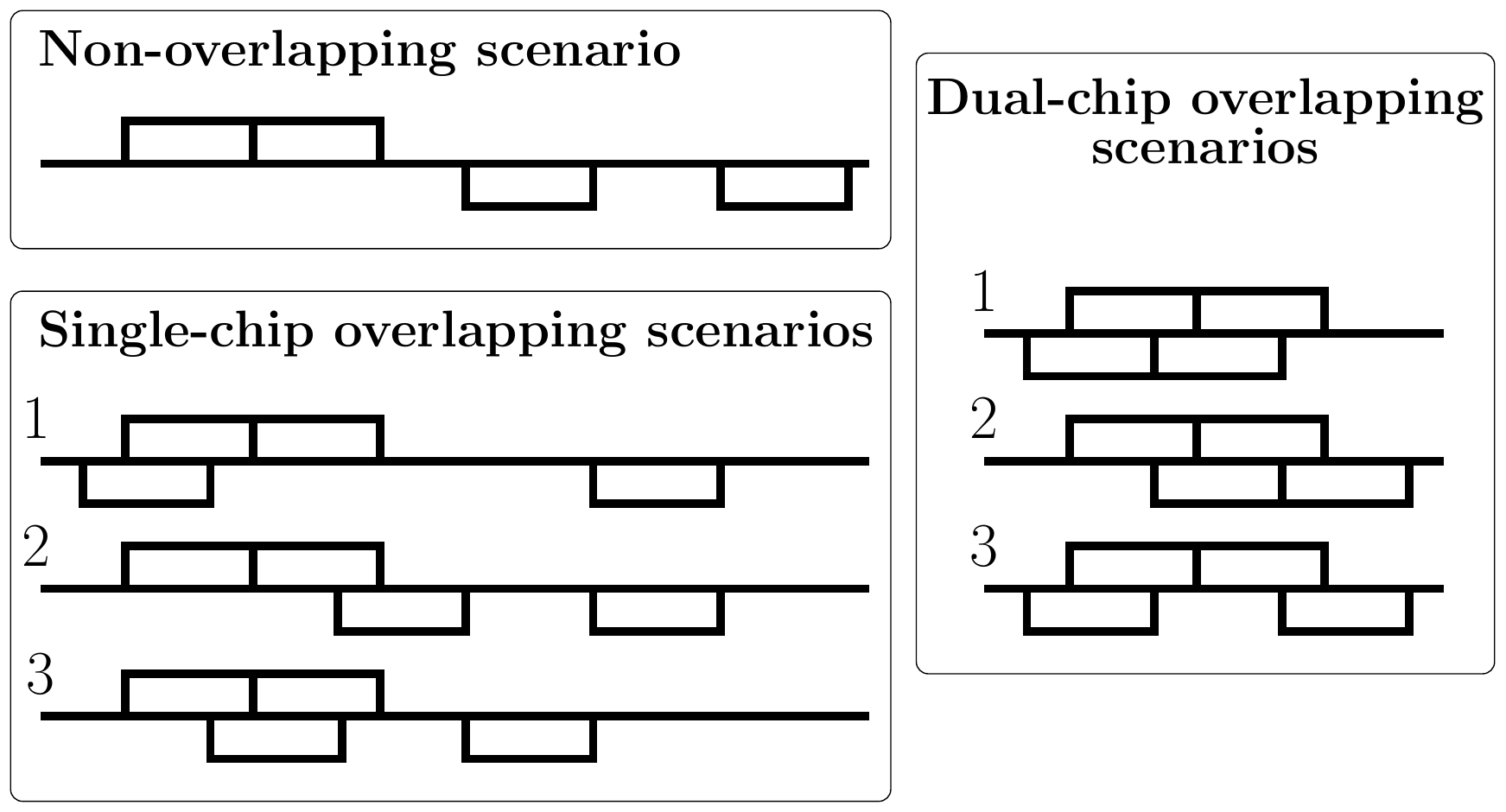}	
		}
	\vspace{-0.3cm}
	\caption{\small Different pattern-overlapping scenarios in the chip-asynchonous scenario. The pairs of on-chips of Links~1 and~2 are the rectangulars above and under the lines, respectively.}
	\vspace{-0.7cm}
	\label{fig:asyn}
\end{figure}

\subsection{BER at the Reader for Backward IT}
We assume that Reader~1 suffers interference when Link~1's transmitted on-chip $C_1$ is partially or entirely overlapped by either one or two of Link~2's on-chips $s_{20}$ and $s_{21}$.
Thus, the BER of coherent detection can be close to the maximum of $0.5$ when the interference occurs (see Sec.~\ref{Sec_IT_backward}).
Calculating the probability that Reader~1 suffers interference and following the similar steps for \eqref{slow_fading_reader_ber_0}, the expected BER for the backward IT is 
\begin{equation} \label{async_reader_ber}
\begin{aligned}
&\Preaderasyn
= \Pbpsk \left( \Pdis{0} +  \Pcon{0} + \frac{1}{2} \left(\Pdis{1-1}+\Pdis{1-2} + \Pdis{2-4} +\Pcon{1-1}+\Pcon{1-2}\right) \right)\\
&+ \!0.5 \times \left(\frac{1}{2} \left(\Pdis{1-1} \!+\!\Pdis{1-2} \!+\! \Pdis{2-4} \!+\!\Pcon{1-1}\!+\!\Pcon{1-2}\right) \!+\! \Pdis{2-1} \!+\!\Pdis{2-2} \!+\!\Pdis{2-3}\!+\!\Pcon{1-3} \!+\!\Pcon{2-1}\!+\!\Pcon{2-2}\!+\!\Pcon{2-3}\right).
\end{aligned}
\end{equation}
Substituting \eqref{prob_asyn} into \eqref{async_reader_ber} gives 
the expected BER for the backward IT as
\begin{equation} \label{Asyn_Reader_BER}
\Preaderasyn= \Pbpsk \frac{(N-3) (N-2)}{ N (N-1)}
+
\frac{2 N - 3}{N(N-1)}.
\end{equation}

\textit{\underline{Insights}:} We can make the following observations using \eqref{Asyn_Reader_BER}:
(i) For the high SNR regime, the BER reduces to $\Preaderasyn \approx \frac{2}{N}$. 
Comparing with the chip-synchronous case, we have 
$
{\Preaderasyn}/{\Preader} \approx 2,
$
which is the \emph{BER deterioration rate} due to the chip asynchronization.
(ii) For the static channel, 
as the sequence length $N\rightarrow \infty$, the asymptotic BER for the FCP scheme is given by
$ 
\Preaderasyn \approx \frac{(N-2)(N-3)}{ N(N-1)} \Pbpsk,
$
thus, comparing with the chip-synchronous case, we have ${\Preaderasyn}/{\Preader} \approx 1$.
While the asymptotic BER for the FCE scheme is given by
$ 
\Preader \approx \frac{2}{N},
$
thus, ${\Preaderasyn}/{\Preader} \approx 2$.
(iii) For Rayleigh fading channel,
as the sequence length $N\rightarrow \infty$, the asymptotic BER for both the FCP and FCE schemes are given by
$ 
\Preaderasyn \approx \frac{(N-2)(N-3)}{ N(N-1)}  \Pbpsk,
$
thus, comparing with the chip-synchronous case, ${\Preaderasyn}/{\Preader} \approx 1$.
\emph{From the above observations, we see that the BER deterioration rate for the backward IT is either $1$ (i.e., no deterioration) or $2$, which is not too significant.}
In the next subsection, we show that this is also comparable to the BER deterioration rate for the forward IT.

\subsection{BER at the Tag for Forward IT}
For tractability, we focus on the large sequence length scenario.
Based on \eqref{prob_asyn}, as the sequence length $N \rightarrow \infty$, 
the dominant terms corresponding to the pattern overlapping scenarios are 
$\Pdis{1-1}$ and $\Pdis{1-2}$, thus, the asymptotic BER for the forward IT is 
\begin{equation} \label{asyn_BER}
\Ptagasyn \approx \Pdis{1-1} \myexpect{}{\Ptagi{1-1}} + \Pdis{1-2} \myexpect{}{\Ptagi{1-2}}
\approx \frac{4}{N} \myexpect{}{\Ptagi{1-1} + \Ptagi{1-2}}.
\end{equation}

Given disjunct pair of on-chips for Link~1, for the single-chip overlapping scenario 1 (see Fig.~\ref{fig:asyn}), we have the following transmission cases which occur with the same probability.

		\underline{Case 1}: The two readers are using the overlapping chip for transmission.
		The harvested energy in the chips $C_1$, $\mathcal{S}_1 \backslash C_1$, the off-chips and the BER for the forward IT~are 
		\begin{equation}
			\begin{aligned}
					\Erx{1-1}{(C_1=C_2=\mathcal{S}_1 \cap \mathcal{S}_2)} &=\frac{T}{N} \left((1-\beta) \Prxb  + \beta \Prxa\right),\ \Erxnew{1-1}{(C_1=C_2=\mathcal{S}_1 \cap \mathcal{S}_2)} = 0,\\
					\EEH{1-1}{(C_1=C_2=\mathcal{S}_1 \cap \mathcal{S}_2)} &= \frac{T}{N}  \beta \Peh,\ \Ptagi{1-1}{(C_1=C_2=\mathcal{S}_1 \cap \mathcal{S}_2)} = 0.
			\end{aligned}
		\end{equation}
		
		\underline{Case 2}: Reader~1 is using the non-overlapping chip, while Reader~2 is using the overlapping chip for transmission.
		\begin{equation}
		\begin{aligned}
		\Erx{1-1}{(C_1\neq C_2=\mathcal{S}_1 \cap \mathcal{S}_2)} &=\frac{T}{N} \Prxa,\ \Erxnew{1-1}{(C_1\neq C_2=\mathcal{S}_1 \cap \mathcal{S}_2)} = \frac{T}{N}(1-\beta)\Prxc,\\
		\EEH{1-1}{(C_1\neq C_2=\mathcal{S}_1 \cap \mathcal{S}_2)} &= \frac{T}{N}  \beta \Peh,\ \Ptagi{1-1}{(C_1\neq C_2=\mathcal{S}_1 \cap \mathcal{S}_2)} =\myprobability{\Prxa < (1-\beta) \Prxc}.
		\end{aligned}
		\end{equation}

		\underline{Case 3}: Reader~1 is using the overlapping chip, while Reader~2 is using the non-overlapping chip for transmission.
		\begin{equation}
		\begin{aligned}
		\Erx{1-1}{(C_2\neq C_1=\mathcal{S}_1 \cap \mathcal{S}_2)} &=\frac{T}{N} \left( \beta \Prxa  + (1-\beta)\Prxd\right),\ \Erxnew{1-1}{(C_2\neq C_1=\mathcal{S}_1 \cap \mathcal{S}_2)} = 0,\\
		\EEH{1-1}{(C_2\neq C_1=\mathcal{S}_1 \cap \mathcal{S}_2)} &= \frac{T}{N}  \Peh,\ \Ptagi{1-1}{(C_2\neq C_1=\mathcal{S}_1 \cap \mathcal{S}_2)} =0.
		\end{aligned}
		\end{equation}
	
		\underline{Case 4}: Reader~1 and Reader~2 are using non-overlapping chips for transmission.
		\begin{equation}
		\begin{aligned}
		\Erx{1-1}{(C_1\neq \mathcal{S}_1 \cap \mathcal{S}_2, C_2\neq \mathcal{S}_1 \cap \mathcal{S}_2)} &=\frac{T}{N} \Prxa,\ \Erxnew{1-1}{(C_1\neq \mathcal{S}_1 \cap \mathcal{S}_2, C_2\neq \mathcal{S}_1 \cap \mathcal{S}_2)} = 0,\\
		\EEH{1-1}{(C_1\neq \mathcal{S}_1 \cap \mathcal{S}_2, C_2\neq \mathcal{S}_1 \cap \mathcal{S}_2)} &= \frac{T}{N}  \Peh,\ \Ptagi{1-1}{(C_1\neq \mathcal{S}_1 \cap \mathcal{S}_2, C_2\neq \mathcal{S}_1 \cap \mathcal{S}_2)} =0.
		\end{aligned}
		\end{equation}

Since the only difference between the single-chip overlapping scenarios~1 and~2 is the delay offset, which means by replacing $\beta$ with $1-\beta$ in the above analysis, the relevant results for the single-chip overlapping scenario~2 can be obtained.
Thus, based on \eqref{asyn_BER}, 
as the sequence length $N\rightarrow \infty$, the asymptotic BER for the forward IT is 
\begin{equation} \label{asyn_tag_ber_gen}
\Ptagasyn = \frac{1}{N} 
\left(\myprobability{\Prxa < (1-\beta) \Prxc} + \myprobability{\Prxa < \beta \Prxc}\right).
\end{equation}

For the static channel, without loss of generality, assuming that $\beta\in (0,1/2)$,
if the random variable $\Prxa/\Prxc \in [0,\beta)$ with probability $1$, $\Ptagasyn = \frac{2}{N} > \Ptag = \frac{1}{N}$,
whereas if the random variable $\Prxa/\Prxc \in (1-\beta,1)$ with probability $1$, $\Ptagasyn = o(\frac{1}{N}) < \Ptag = \frac{1}{N}$.
Thus, for some cases, the chip asynchronization deteriorates BER, but not for other cases.

For Rayleigh fading channel, based on \eqref{asyn_tag_ber_gen}, we further have
\begin{equation} \label{asyn_approx_ber}
\begin{aligned}
\Ptagasyn = \frac{1}{N} 
\left(
2- \int_{0}^{\infty} \left(
\frac{e^{-x}}{1+ \beta d^{\lambda}_{11} \left(\frac{1}{d^{\lambda}_{21}} + \frac{\rho x}{d^{\lambda}_{21} d^{\lambda}_{t}} \right)}
+
\frac{e^{-x}}{1+ (1-\beta) d^{\lambda}_{11} \left(\frac{1}{d^{\lambda}_{21}} + \frac{\rho x}{d^{\lambda}_{21} d^{\lambda}_{t}} \right)}
\right) \mathrm{d}x
\right),
\end{aligned}
\end{equation}
and $\left.{\Ptagasyn}\right\vert_{\beta = 0} \!=\!\left.{\Ptagasyn}\right\vert_{\beta = 1} $, 
$\left.\frac{\mathrm{d}\Ptagasyn}{\mathrm{d} \beta}\right\vert_{\beta = 0}\!>\!0$, $\left.\frac{\mathrm{d}\Ptagasyn}{\mathrm{d} \beta}\right\vert_{\beta = 1}\!< \!0$,  $\left.\frac{\mathrm{d}\Ptagasyn}{\mathrm{d} \beta}\right\vert_{\beta = \frac{1}{2}}\!= \!0$  and $\frac{\mathrm{d}^2\Ptagasyn}{\mathrm{d}^2 \beta}\!<\!0$.
Thus, $\Ptagasyn$ is a concave function of $\beta$, which increases first and then decreases.
Therefore, any chip asynchronization deteriorates the BER, and the worst case of BER is obtained when $\beta = \frac{1}{2}$ as
\begin{equation} \label{asyn_tag_ber}
\Ptagasyn = \frac{2}{N} 
\left(
1- \int_{0}^{\infty} \left(
\frac{e^{-x}}{1+ \frac{1}{2} d^{\lambda}_{11} \left(\frac{1}{d^{\lambda}_{21}} + \frac{\rho x}{d^{\lambda}_{21} d^{\lambda}_{t}} \right)}	
\right) \mathrm{d}x
\right),
\end{equation}
which can be proved to be greater than $1/N$ but less than $2/N$.

\textit{\underline{Insight}:} {Therefore, the BER deterioration rate for the forward IT due to the chip asynchronization, $\Ptagasyn/\Ptag \in \left[1,2\right]$, when the TH-SS sequence length is sufficiently large.}

\subsection{Performance of Forward ET}
\subsubsection{Expected ETR}
Based on the analysis above for the single chip overlapping scenario~1 given Link 1's disjunct pair of on-chips, we see that the expected ETR in this scenario~is 
\begin{equation} \label{asyn_energy_1}
\myexpect{}{\Erx{1-1}{} + \Erxnew{1-1}{} + \EEH{1-1}{}}
= \frac{T}{4 N} \left((2+2 \beta) \left(\Prxa + \Peh\right) + (1-\beta) \left(\Prxb + \Prxc + \Prxd\right) \right).
\end{equation}
Similarly, for the single chip overlapping scenario~2, we have
\begin{equation} \label{asyn_energy_2}
\myexpect{}{\Erx{1-2}{} + \Erxnew{1-2}{} + \EEH{1-2}{}}
= \frac{T}{4 N} \left((2+2 (1-\beta)) \left(\Prxa + \Peh\right) + \beta \left(\Prxb + \Prxc + \Prxd\right) \right),
\end{equation}
and thus, 
the expected ETR in the single-chip overlapping scenario given given Link 1's disjunct pair of on-chips, is the expectation of \eqref{asyn_energy_1} and \eqref{asyn_energy_2}, i.e., 
$
\frac{T}{4 N} \left(3 \left(\Prxa \!+ \Peh\right) \!+ \frac{1}{2} \left(\Prxb \!+ \Prxc \!+ \Prxd\right) \right),
$
which is independent of the delay offset $\beta$.
Similarly, for the other pattern overlapping scenarios, the expected ETR also do not rely on $\beta$, thus, we have
\begin{equation}\label{asyn_average_energy}
\Etagasyn = \Etag = \frac{T}{N} \myexpect{}{\left(
	\frac{(N-2)}{N} \Prxa
	+ \frac{1}{N} \left(\Prxb + \Prxc +\Prxd  \right) 
	+ \frac{(N-2)}{N} \Peh
	\right)}.
\end{equation}
Intuitively, the TH-SS scheme has averaged out the delay offset effect on the expected ETR. 
Thus, the chip asynchronization has zero effect on expected ETR.

\subsubsection{Energy-Outage Probability}
As the sequence length $N \rightarrow \infty$, for the FCP scheme, it is straightforward that the asymptotic energy-outage probability $\Poutasyn \rightarrow 1$.
While for the FCE scheme, focusing on the dominant term, the asymptotic energy-outage probability is 
\begin{equation}
\Poutasyn \approx \Pdis{0} \myprobability{\frac{T}{N} \left( \Prxa+\Peh\right) < \mathcal{E}_0}
\approx \myprobability{\frac{T}{N} \left( \Prxa+\Peh\right) < \mathcal{E}_0}.
\end{equation}
Thus, $\Poutasyn/\Pout \approx 1$.
In other words, when the sequence length is sufficiently large, the chip asynchronization effect on the energy-outage probability is negligible.

\textit{\underline{Insight}:} {Therefore, as the sequence length $N \rightarrow \infty$, the chip asynchronization has negligible effect on the performance of the forward ET.}

\section{Performance of Time-Hopping Full-Duplex BackCom: $K$-Link Case} \label{Sec_K}
We study the $K$-link chip-synchronous transmissions in this section.
Assuming that $\varrho_n$, $n=0,1,2$, is the probability that Link~1 has $n$ chips overlapped by the other links,
i.e., $n = \vert \mathcal{S}_1 \cap (\mathcal{S}_2 \cup \mathcal{S}_3\cup \cdots \cup \mathcal{S}_K) \vert$, 
we can obtain 
\begin{equation} \label{K_prob}
\begin{aligned}
\varrho_0 &=\! p^{K-1}_0 = \left( \frac{(N-2)(N-3)}{N(N-1)} \right)^{K-1} = \mathcal{O}(1),\\
\varrho_1 &=\! 2 \left(\!\left(\!p_0 \!+\! \frac{1}{2} p_1\!\right)^{K-1} \!-\! p_0^{K-1}\!\right) \!=\! 2 \left(\!\left(\!\frac{N\!-\!2}{N}\!\right)^{K-1} \!-\! \left(\! \frac{(N\!-\!2)(N\!-\!3)}{N(N\!-\!1)} \!\right)^{K\!-\!1}\!\right) \!=\! \mathcal{O}(\frac{1}{N}),\\
\varrho_2 &=\! 1- \varrho_0 - \varrho_1 = 1 +  \left( \frac{(N-2)(N-3)}{N(N-1)} \right)^{K-1} - 2 \left(\frac{N-2}{N}\right)^{K-1} = o(\frac{1}{N}),
\end{aligned}
\end{equation}
where $\mathcal{O}(\cdot)$ and $o(\cdot)$ are the big O and little o notations, respectively.

\subsection{BER at the Reader for Backward IT}
Since Reader~1 suffers interference only if Link~1's transmitted on-chip $C_1 \in \mathcal{S}_2 \cup \mathcal{S}_3, \cdots \cup \mathcal{S}_K$,
following the similar steps for \eqref{slow_fading_reader_ber_0}, i.e., replacing $p_n$ with $\varrho_n$ in \eqref{slow_fading_reader_ber_0}, the expected BER for the backward IT is
\begin{equation}  \label{K_BER_reader}
\begin{aligned}
\Preader
& = \Pbpsk \left(\frac{N-2}{N}\right)^{K-1}
+ 
\frac{1}{2}\left(1-\left(\frac{N-2}{N}\right)^{K-1}\right).
\end{aligned}
\end{equation}

\textit{\underline{Insight}:} For the high SNR regime, the BER reduces to $\Preader \approx \frac{K-1}{N}$, and thus, \emph{$\Preader$ increases linearly with the number of BackCom Links, $K$.}
%
%
%

For the static channel, 
as the sequence length $N\rightarrow \infty$, the asymptotic BER for the FCP scheme is given by
$
\Preader \approx \left(\frac{N-2}{N}\right)^{K-1} \Pbpsk.
$
While the asymptotic BER for the FCE scheme is given by
$
\Preader \approx  \frac{K-1}{N},
$
which \emph{increases linearly with the number of BackCom Links, $K$.}
For Rayleigh fading channel,
as the sequence length $N\rightarrow \infty$, the asymptotic BER for both the FCP and FCE schemes are given by
$
\Preader \approx \left(\frac{N-2}{N}\right)^{K-1}  \Pbpsk
$,
i.e., \emph{increasing the number of the BackCom links has negligible effect on the BER for the backward IT.}

\subsection{BER at the Tag for Forward IT}
Based on Sec.~\ref{Sec_IT_forward}, Tag~1 suffers interference only in the single-chip or dual-chip pattern-overlapping scenarios. 
Assuming $\varrho$ is the probability that only one chip of the pattern $\mathcal{S}_1$ is overlapped and the overlapping is caused by just one link, we have 
\begin{equation}
\varrho = 2 \left( (K-1) \frac{p_1}{2} (p_0)^{K-2} \right) = (K-1) \frac{4(N-2)}{N(N-1)} \left(\frac{(N-2)(N-3)}{N(N-1)}\right)^{K-2} = \mathcal{O}(\frac{1}{N}),
\end{equation}
and it is easy to see that the sum probability of all the other single-chip overlapping and dual-chip overlapping scenarios is $\varrho_1+\varrho_2 - \varrho$, which approaches zero with a higher order of $\frac{1}{N}$.
Therefore, in the large sequence length regime, when analyzing the BER for the forward IT, we only consider the pre-mentioned dominant case.

Thus, following the similar steps for \eqref{first_ave_ber}, Tag~1's expected BER is 
\begin{equation} \label{K_Tag}
\begin{aligned}
\Ptag & \approx \varrho \frac{1}{K-1} \sum\limits_{k=2}^{K} \frac{1}{4} \myprobability{\Prx{1}{\left(C_1 \neq C_k = \mathcal{S}_1 \cap \mathcal{S}_k\right)} > \Prxa}\\
&= \frac{N-2}{N(N-1)} \left(\frac{(N-2)(N-3)}{N(N-1)}\right)^{K-2}  \sum\limits_{k=2}^{K}  \myprobability{\Prx{1}{\left(C_1 \neq C_k = \mathcal{S}_1 \cap \mathcal{S}_k\right)} > \Prxa},
\end{aligned}
\end{equation}
where $\Prx{1}{\left(C_1 \neq C_k = \mathcal{S}_1 \cap \mathcal{S}_k\right)} \triangleq  \etabs \myP \vert f_{k1} u + f_{kk}\sqrt{\rho} g_{k1} q_k \vert^2 $ which is the combined receive signal power of the signals R$k$-T1 and R$k$-T$k$-T1 during the chip $\mathcal{S}_1 \backslash C_1$.

As the sequence length $N\rightarrow \infty$, the asymptotic BER is
\begin{equation}
\Ptag \approx \frac{1}{N} \sum\limits_{k=2}^{K} \myprobability{\Prx{1}{\left(C_1 \neq C_k = \mathcal{S}_1 \cap \mathcal{S}_k\right)} > \Prxa}.
\end{equation}

\textit{\underline{Insight}:} {Therefore, the BER for the forward IT increases with the number of BackCom links~$K$.}

\subsection{Performance of Forward ET}
For simplicity, we focus on the large sequence length regime, and it is easy to see the probability that there is no overlapping between any of the $K$-link patterns, is 
\begin{equation}
\myprobability{\!\vert \cup^{K}_{k=1} \mathcal{S}_k \vert = 2 K\!}
\!=\!
\frac{\frac{N(N-1)}{2} \frac{(N-2)(N-3)}{2} \cdots \frac{(N-2(K-1))(N-2(K-1) -1)}{2} }{(\frac{N(N-1)}{2})^{K}}
\!=\! \mathcal{O}(1).
\end{equation}

Thus, as the sequence length $N \rightarrow \infty$, the asymptotic expected ETR and energy-outage probability are give by 
\begin{align} 
\Etag &\approx \frac{T}{N}   \myexpect{}{\Prxa + \sum\limits_{k=2}^{K} \PehK{k}} \\
\label{K_P_out}
\Pout &\approx 
\myprobability{\!\vert \cup^{K}_{k=1} \mathcal{S}_k \vert = 2 K\!}
\myprobability{\frac{T}{N}\left(\Prxa \!+\! \sum\limits_{k=2}^{K} \!\PehK{k}\right) \!<\! \mathcal{E}_0}
\!\approx\!
\myprobability{\frac{T}{N}\left(\Prxa \!+\! \sum\limits_{k=2}^{K} \!\PehK{k}\right) \!<\! \mathcal{E}_0},
\end{align}
respectively, 
where $\PehK{k} \triangleq \eta \myP \vert f_{k1}+f_{kk}g_{k1} \vert^2$ is the receive signal power of the signal R$k$-T1 in Tag~1's off-chips.
The expected ETR and the energy-outage probability monotonically increases and decreases with $K$, respectively. 

\textit{\underline{Insight}:} {Therefore, when the sequence length is sufficiently large, increasing the number of BackCom Links improves the performance of the forward ET.}

\section{Numerical Results}
In this section, due to space limitations, we only investigate the Rayleigh fading channel scenario (applicable to future IoT) and the FCP scheme.
In general, we focus on the $2$-link case, i.e., $K=2$. 
We set the path loss exponent as $\lambda = 2.5$, and the distance (in meters) between the devices as $d_{11} = 10$, $d_{22} = 10$, $d_{t} = 20$, $d_{12} = 22$, $d_{21} = 22$. 
For the $K$-link case, i.e., $K>2$, we assume that $d_{ii} = d_{11}$, and $d_{i1} = d_{21}$, $i=3,4,...,K$.
Also we set the noise variance at the reader and the tag as $\sigma^2_{\text{reader}} =\sigma^2_{\text{tag}} = -100$ dBm, and the RF energy harvesting efficiency as $\eta = 0.5$~\cite{Huang15} and~\cite{Bi15}.
Unless otherwise stated, we set the reader transmit power as $\myP = 50$ mW, the sequence length as $N=1000$, and the tag power consumption as $\frac{\mathcal{E}_0}{T} = 0.01$~$\mu$W.

{
In the following, we plot the BERs for the forward and backward IT (i.e., $\Preader$ and $\Ptag$), and the energy-outage probability (i.e., $\Pout$) for the forward ET based on the analytical results derived in Sections.~IV-VII. 
The Monte Carlo simulation results, averaged over $10^9$ random channel realizations, are also presented.
Specifically, the analytical results of BER for the forward and backward IT in Sec.~\ref{Sec_IT} are based on two modeling assumptions, respectively, i.e., the BER at the reader is $0.5$ as long as the reader suffers from interference, and 
the BER at the tag does not take into account the noise effect.
These modeling assumptions are verified by the simulation results.

\begin{figure*}[t]
	\renewcommand{\captionlabeldelim}{ }	
	\renewcommand{\captionfont}{\small} \renewcommand{\captionlabelfont}{\small}
	\minipage{0.45\textwidth}
	\centering
		\vspace*{-0.8cm}	
	\includegraphics[scale=0.46]{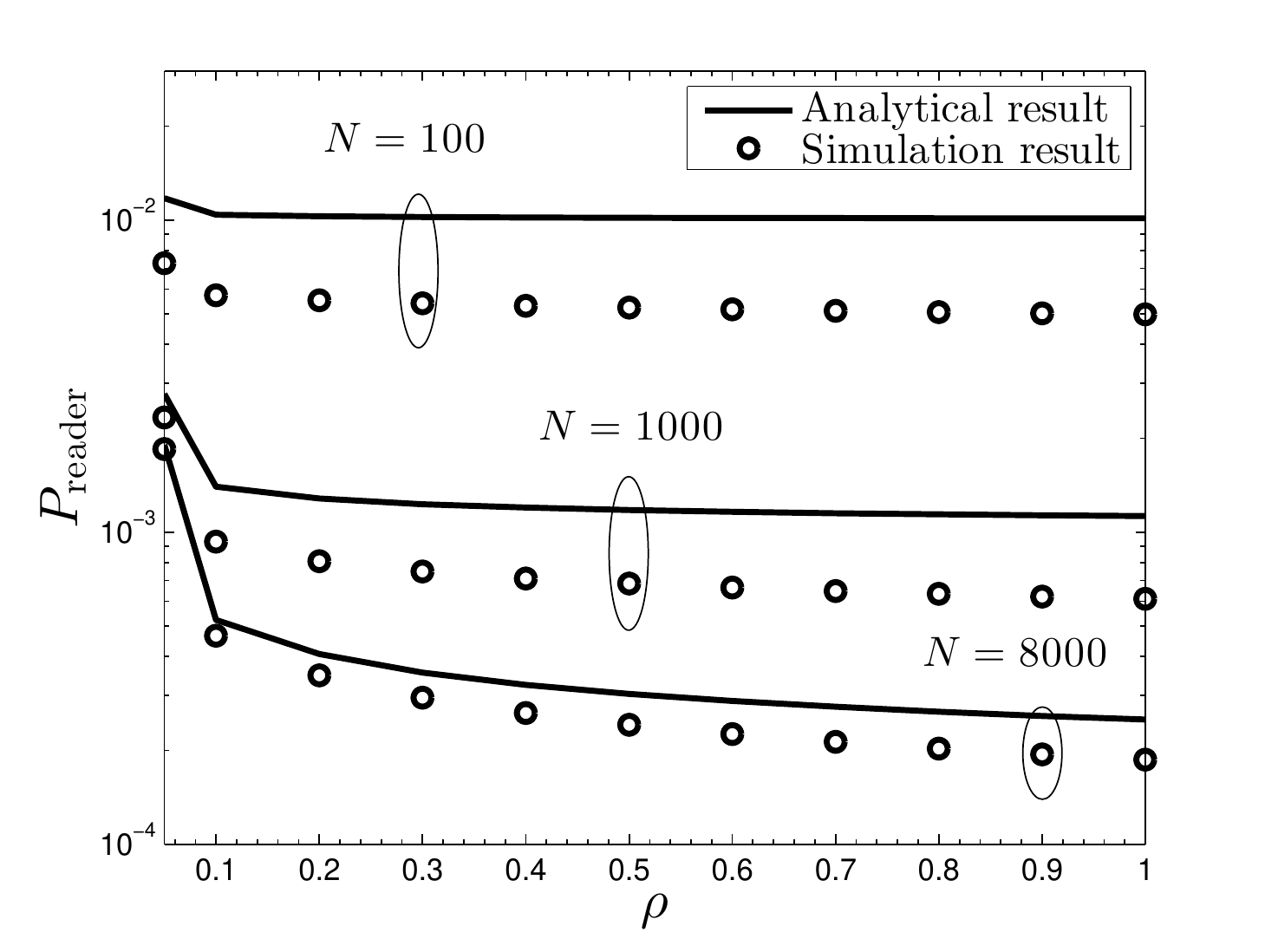}	
	\vspace*{-0.6cm}
	\caption{\small $\Preader$ versus $\rho$.}\label{fig:Preader}
	\endminipage
	\hspace{1cm}
	\minipage{0.45\textwidth}	
	\centering
		\vspace*{-0.8cm}
	\includegraphics[scale=0.46]{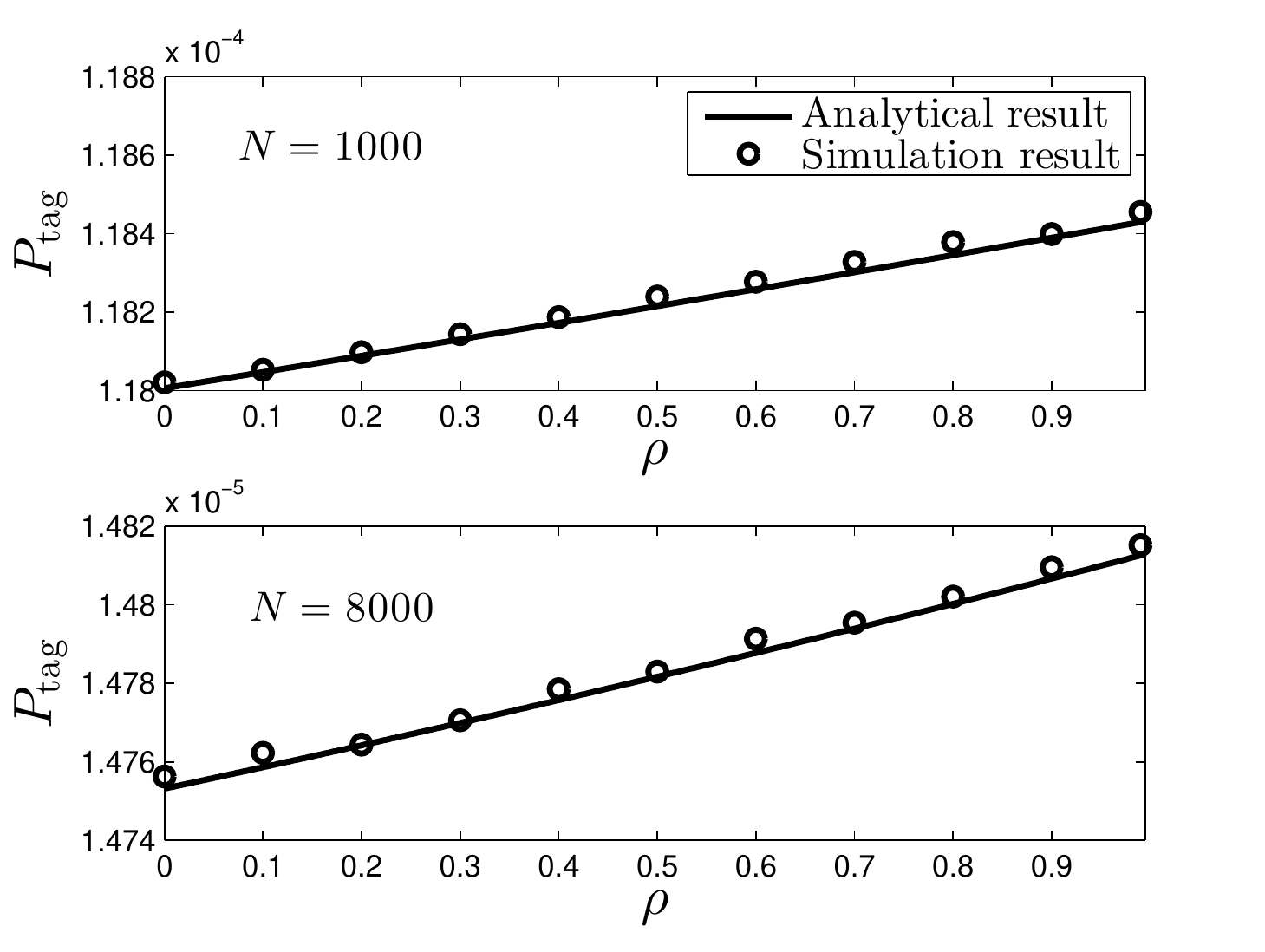}	
	\vspace*{-0.6cm}
	\caption{\small $\Ptag$ versus $\rho$.}\label{fig:Ptag}
	\endminipage
	\vspace*{-0.7cm}
\end{figure*}

\subsection{Two-Link BackCom}

In Fig.~\ref{fig:Preader}, curves of the BER for the backward IT, $\Preader$, are plotted for different reflection coefficient, $\rho$, and TH-SS sequence length, $N$. The analytical result is plotted using Proposition~\ref{proposition_1}.
We see that the analytical result is an upper bound of the simulation result, and the gap diminishes as $N$ increases, e.g., the gap is less than $10^{-4}$ when $N=8000$. Thus, although the inter-link interference may not be fatal, i.e., inducing a BER of $0.5$ at the reader, the analytical result is a tight upper bound especially when $N$ is large.

In Fig.~\ref{fig:Ptag}, curves of the BER for the forward IT, $\Ptag$, are plotted for different $\rho$ and $N$. The analytical result is plotted using Proposition~\ref{Tag_BER_fading}.
We see that the analytical results perfectly match the simulation results, which verifies that the noise effect of the forward BER is negligible under the practical settings.

In Fig.~\ref{fig:Pout}, curves of the energy-outage probability for the forward ET, $\Pout$, are plotted for different $\rho$ and $N$. The analytical result is plotted using Proposition~\ref{fading_outage}, which perfectly matches the simulation result. Hence, the accuracy of the analytical result is verified.

From Figs.~\ref{fig:Preader}-\ref{fig:Pout}, we see that $\Preader$ decreases while both $\Ptag$ and $\Pout$ increase with increasing $\rho$.
This is because a larger $\rho$ induces a stronger backscattered signal and a weaker received signal at the tag, which enhances the backward SNR but reduces the performance of the forward IT and ET.
Also we see that both $\Preader$ and $\Ptag$ decrease while $\Pout$ increases with increasing $N$.
This is because a larger sequence length $N$ suppresses the interference for both the forward and backward IT by reducing the pattern-overlapping probability.
However, it makes the readers have a shorter time (i.e., each chip has a shorter time when $N$ is larger) for active transmissions, which reduces the performance of the forward ET.

Therefore, for practical BackCom system design, the tradeoff between the backward IT and the forward IT/ET with reflection coefficient, and the tradeoff between the IT and the ET with TH-SS sequence length should be carefully considered. Moreover, the reflection coefficient and the sequence length should be optimized to satisfy the performance requirement of a certain BackCom system.

\begin{figure*}[t]
	\renewcommand{\captionlabeldelim}{ }	
	\renewcommand{\captionfont}{\small} \renewcommand{\captionlabelfont}{\small}
	\minipage{0.45\textwidth}
	\centering
		\vspace*{-0.8cm}	
	\includegraphics[scale=0.46]{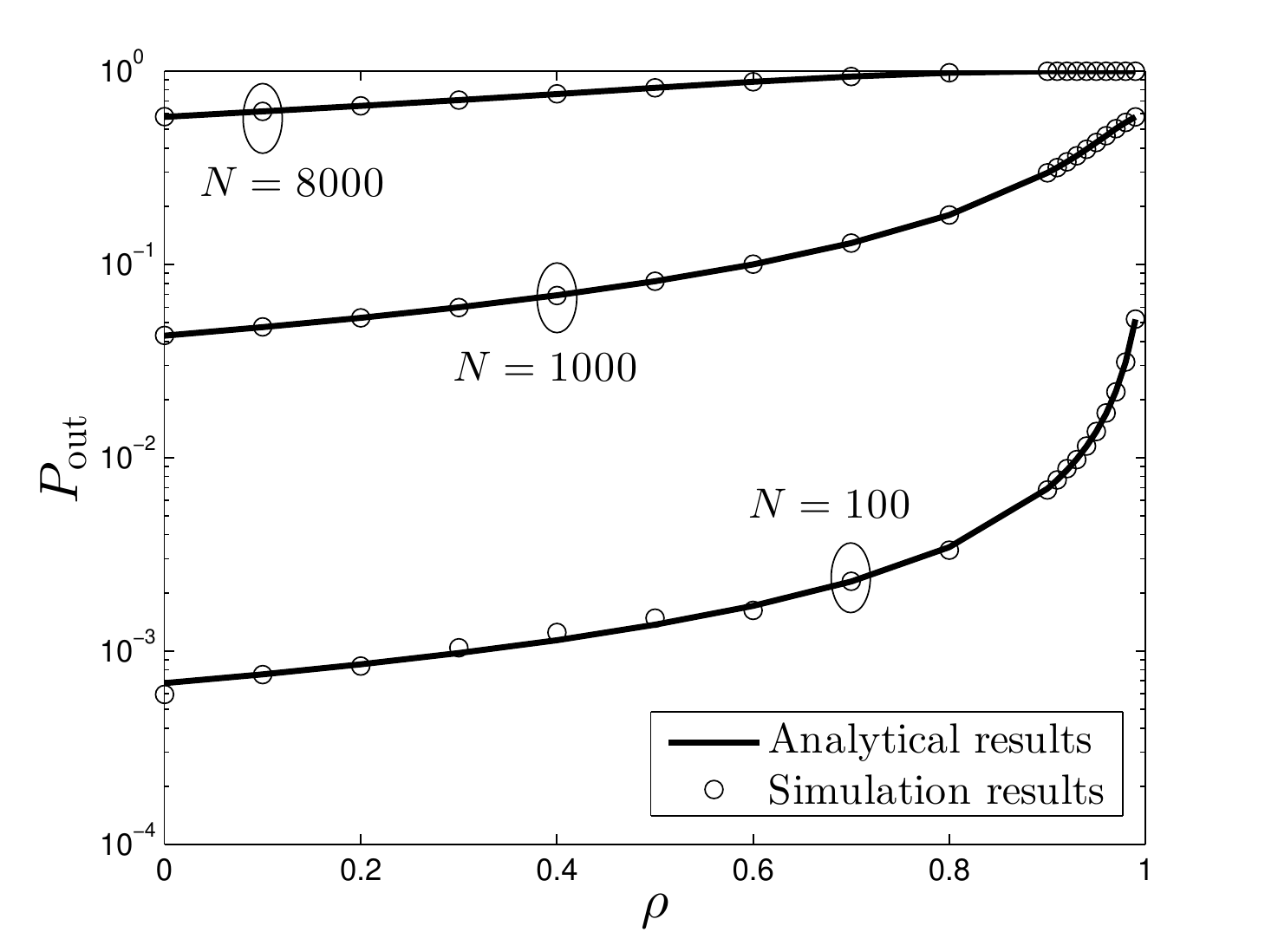}	
	\vspace*{-0.6cm}
	\caption{\small $\Pout$ versus $\rho$.}\label{fig:Pout}
	\endminipage
	\hspace{1cm}
	\minipage{0.45\textwidth}	
	\centering
		\vspace*{-0.8cm}
	\includegraphics[scale=0.46]{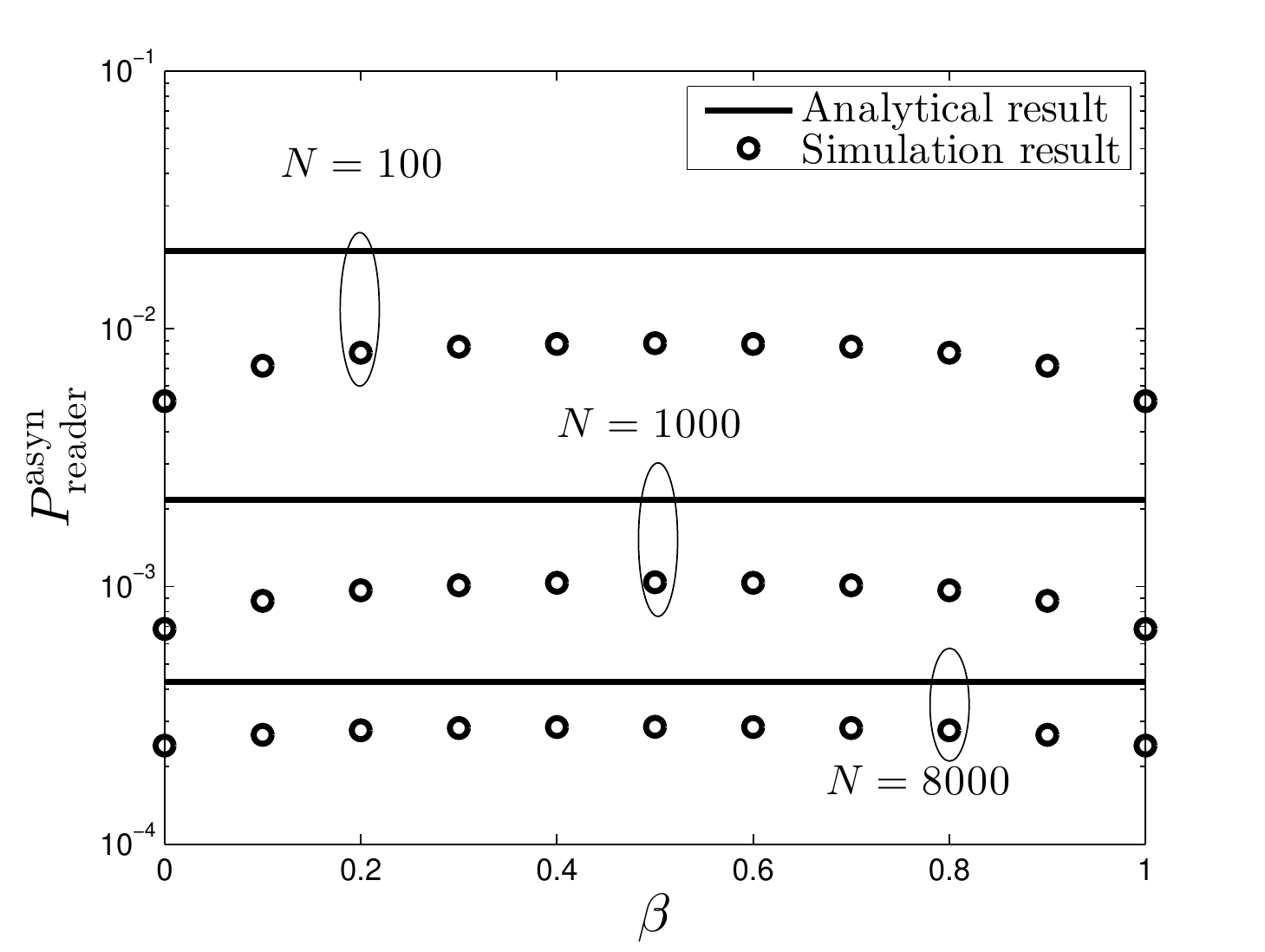}	
	\vspace*{-0.6cm}
	\caption{\small $\Preaderasyn$ versus $\beta$.}\label{fig:Asyn_Reader}
	\endminipage
		\vspace*{-0.7cm}
\end{figure*}

\subsection{Asynchronous BackCom}
%
In Fig.~\ref{fig:Asyn_Reader}, curves of the BER for the backward IT, $\Preaderasyn$, are plotted for different delay offset, $\beta$, and sequence length, $N$. The analytical result is plotted using \eqref{Asyn_Reader_BER}.
We see that the analytical result is an upper bound of the simulation result, and the gap diminishes quickly as $N$ increases, e.g., the gap is about $10^{-3}$ when $N=1000$, and is about $10^{-4}$ when $N=8000$. 
From the simulation result, we see that the BER for the backward IT is mostly affected when the delay offset caused by chip asynchronization is equal to a half chip duration. The influence on the BER caused by asynchronous transmissions is negligible when $N$ is sufficiently large, i.e., the BER is almost the same with $\beta = 0$ and $0.5$, when $N=8000$.
Therefore, although the analytical result is based on the assumption that the BER is the same no matter what the delay offset is, the result is a tight upper bound especially when $N$ is large.

In Fig.~\ref{fig:Asyn_Tag}, curves of the BER for the forward IT, $\Ptagasyn$, are plotted for different $\beta$ and $N$.
The analytical result is plotted using \eqref{asyn_tag_ber}, which matches the simulation result. Hence, the approximation in \eqref{asyn_tag_ber} is tight.
It is observed that the chip asynchronization always increases the BER, i.e., the BER is larger for any $\beta\in(0,1)$ compared with $\beta = 0$ or $1$.
Furthermore, similar with the backward transmission, when the delay offset between the two links is equal to a half chip duration, i.e., $\beta  = 0.5$, the BER is the worst. 
Also it is clear that the chip asynchronization effect on BER can be eliminated by increasing sequence length, for example, the worst chip-asynchronous BER with $N = 8000$ is much smaller than the BER of the chip-synchronous case with $N=1000$.

Therefore, Figs.~\ref{fig:Asyn_Reader} and~\ref{fig:Asyn_Tag} jointly show that the chip-asynchronous effect on the forward and backward IT is negligible as long as the sequence length is sufficiently large.

\begin{figure*}[t]
	\renewcommand{\captionlabeldelim}{ }	
	\renewcommand{\captionfont}{\small} \renewcommand{\captionlabelfont}{\small}
	\minipage{0.45\textwidth}
	\centering
		\vspace*{-0.8cm}	
	\includegraphics[scale=0.46]{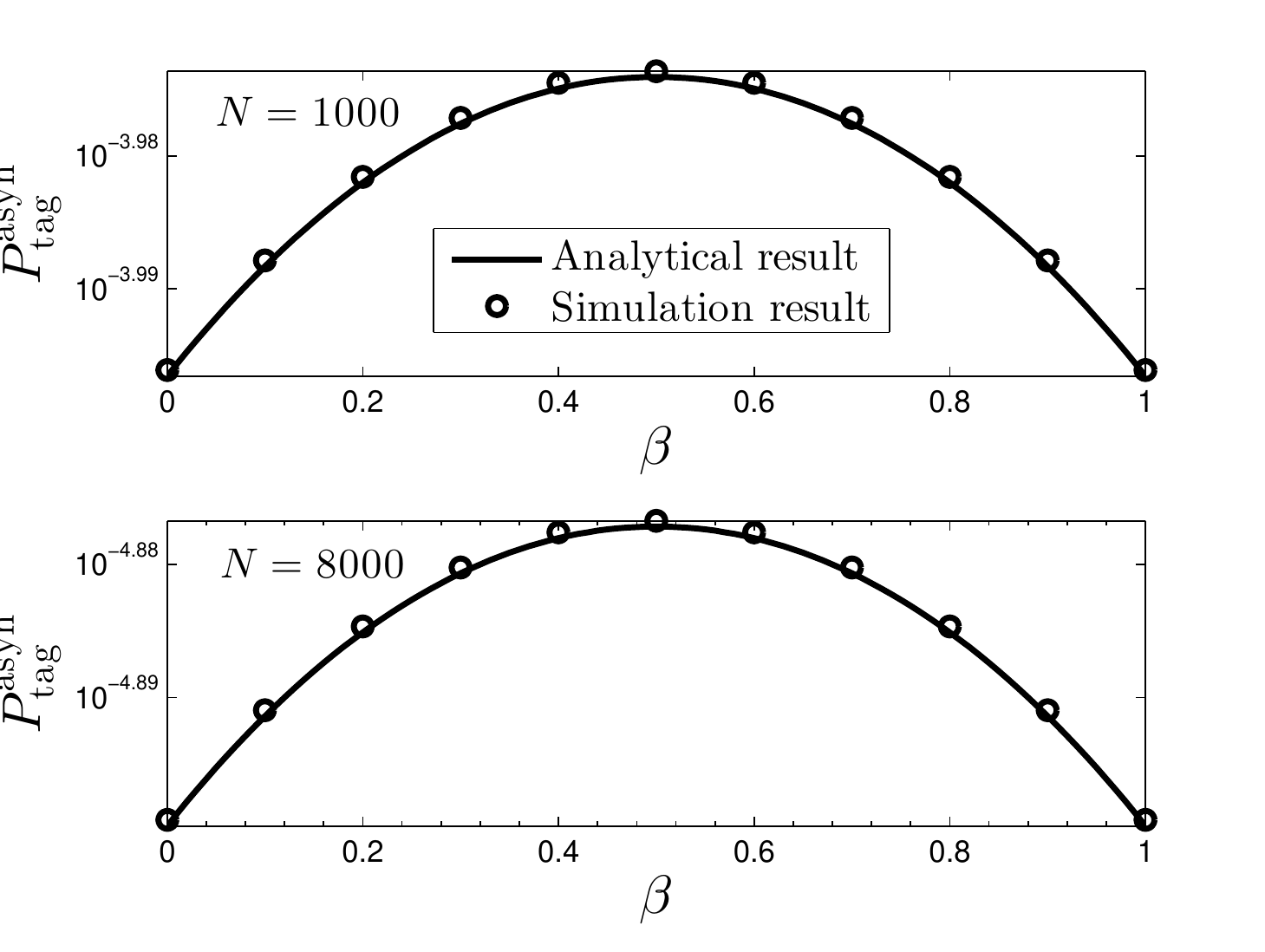}	
	\vspace*{-0.6cm}
	\caption{\small $\Ptagasyn$ versus $\beta$.}\label{fig:Asyn_Tag}
	\endminipage
	\hspace{1cm}
	\minipage{0.45\textwidth}	
	\centering
		\vspace*{-0.8cm}
	\includegraphics[scale=0.46]{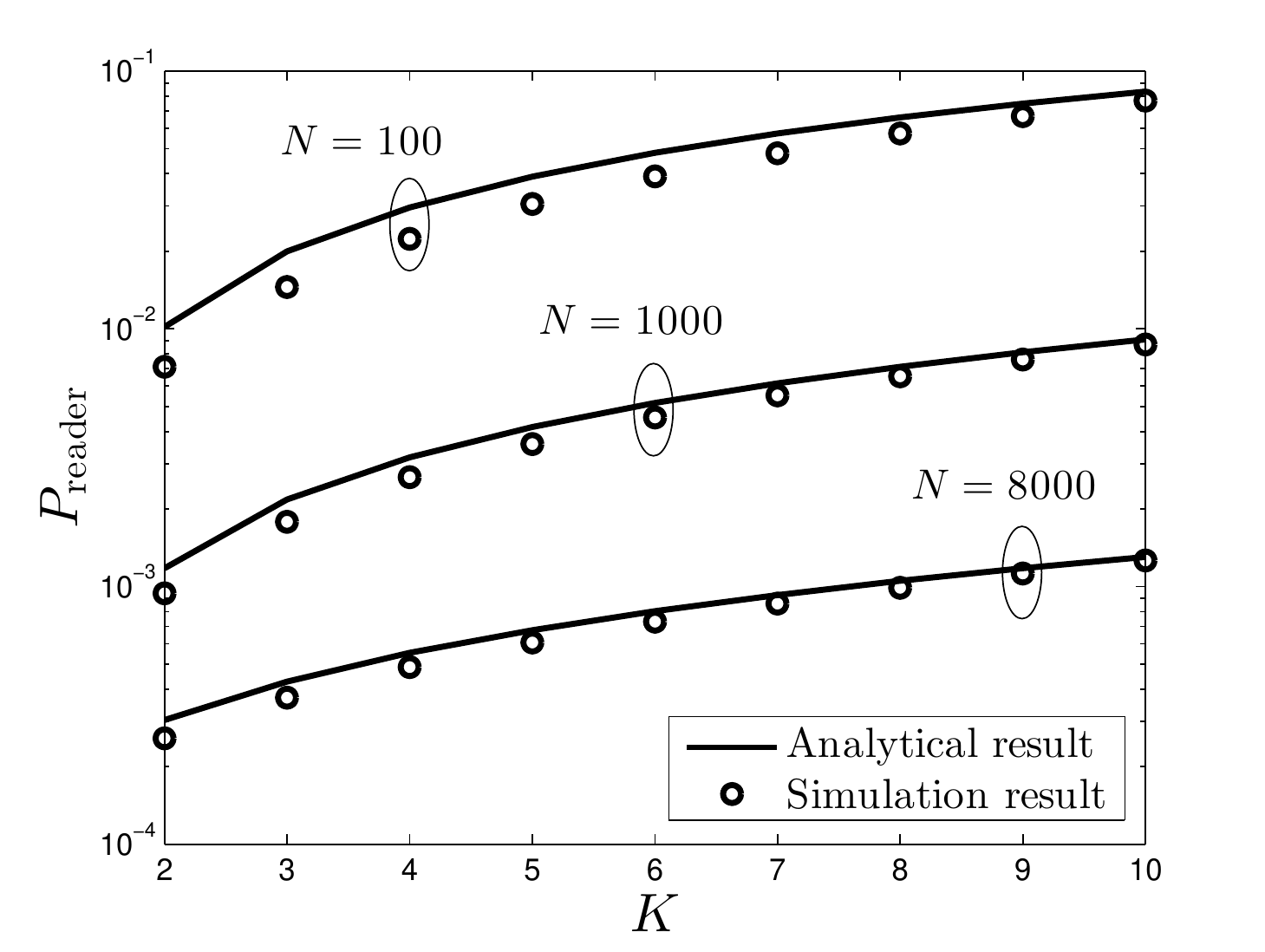}	
	\vspace*{-0.6cm}
	\caption{\small $\Preader$ versus $K$.}\label{fig:K_Reader}
	\endminipage
		\vspace*{-0.7cm}
\end{figure*}

\subsection{$K$-Link BackCom}
%
In Fig.~\ref{fig:K_Reader}, curves of $\Preader$ are plotted for different number of BackCom links, $K$, and different sequence length, $N$. The analytical result is plotted using \eqref{K_BER_reader}, which is an upper bound of the accurate result since it assumes a BER of $0.5$ when the interference occurs at the reader.
We see that the analytical result is a tight upper bound of the simulation result, and the gap diminishes with the increasing sequence length and number of BackCom links. 
Also we see that $\Preader$ increases with $K$, which is mainly because more BackCom links make the backward IT more likely to suffer from interference.
For a fixed $N$, the BER deteriorates as the number of BackCom links becomes large, and can even become close to $0.5$, when the interference occurs.

In Fig.~\ref{fig:K_Tag}, curves of $\Ptag$ are plotted for different $K$ and $N$. The analytical result is plotted using \eqref{K_Tag}, which is a lower bound since it only takes into account the dominant term for a large $N$.
We see that the analytical result is a tight lower bound of the simulation result.
Also we see that $\Ptag$ increases with $K$, since more BackCom links make the backward IT more likely to suffer from interference, and the interference is stronger when it occurs.

In Fig.~\ref{fig:K_Pout}, curves of the energy-outage probability, $\Pout$, are plotted for different $K$ and $N$. The analytical result is plotted using \eqref{K_P_out}, which is a lower bound since it only takes into account the dominant term for a large $N$.
We see that the analytical result is a tight lower bound of the simulation result. 
Also we see that $\Pout$ decreases with $K$, since a large number of BackCom links increases the harvested energy at the tag due to the multi-reader transmissions.

Figs.~\ref{fig:K_Reader}-\ref{fig:K_Pout} jointly show the performance tradeoff between the IT and the ET with the number of BackCom links. 
Thus, for practical BackCom network design, this tradeoff should be carefully considered, and the number of BackCom links should be optimized to satisfy both the performance requirement of IT and ET.
}

\begin{figure*}[t]
	\renewcommand{\captionlabeldelim}{ }	
	\renewcommand{\captionfont}{\small} \renewcommand{\captionlabelfont}{\small}
	\minipage{0.45\textwidth}
	\centering
		\vspace*{-0.8cm}	
	\includegraphics[scale=0.46]{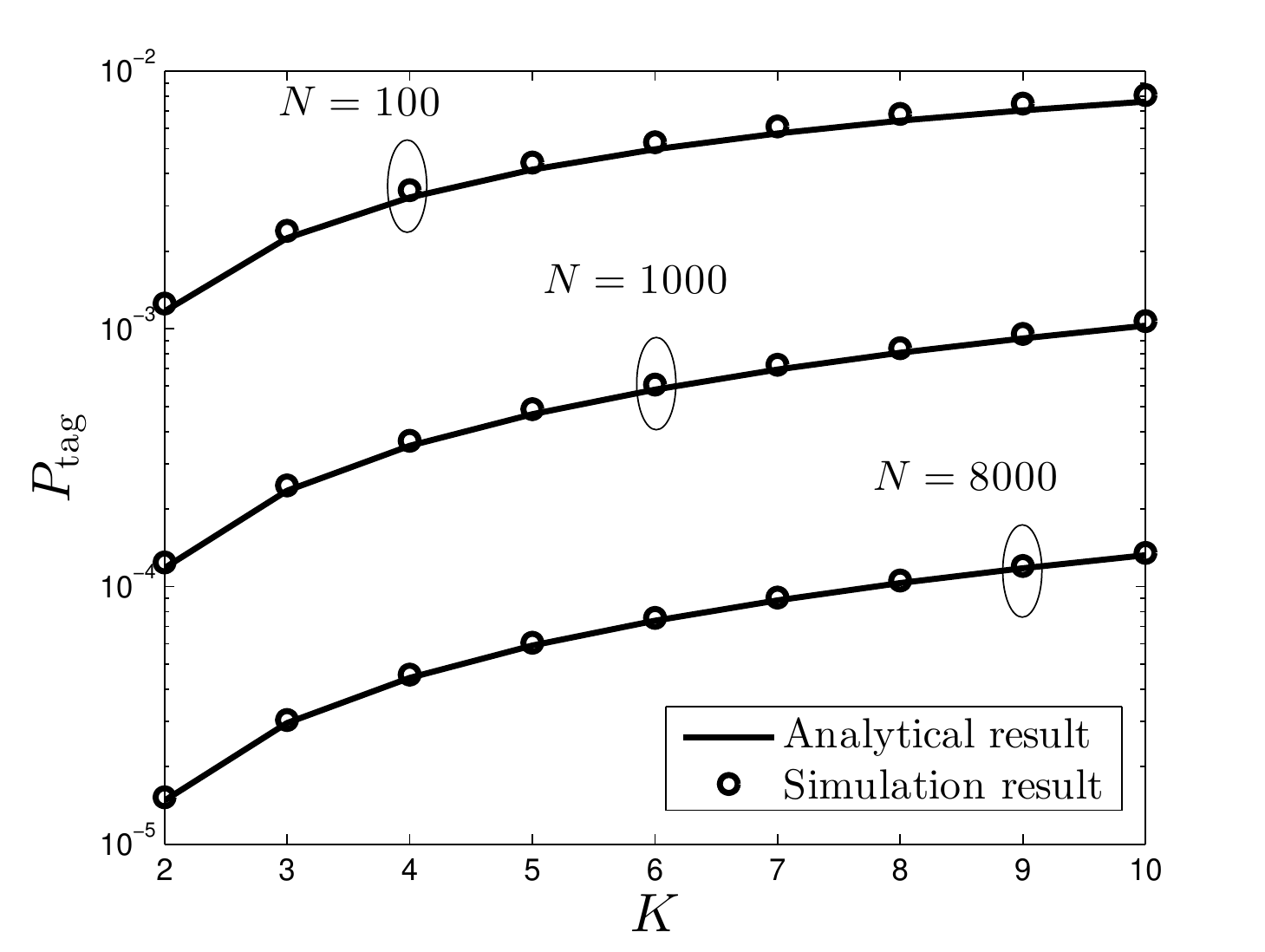}	
	\vspace*{-0.6cm}
	\caption{\small $\Ptag$ versus $K$.}\label{fig:K_Tag}
	\endminipage
	\hspace{1cm}
	\minipage{0.45\textwidth}	
	\centering
		\vspace*{-0.8cm}
	\includegraphics[scale=0.46]{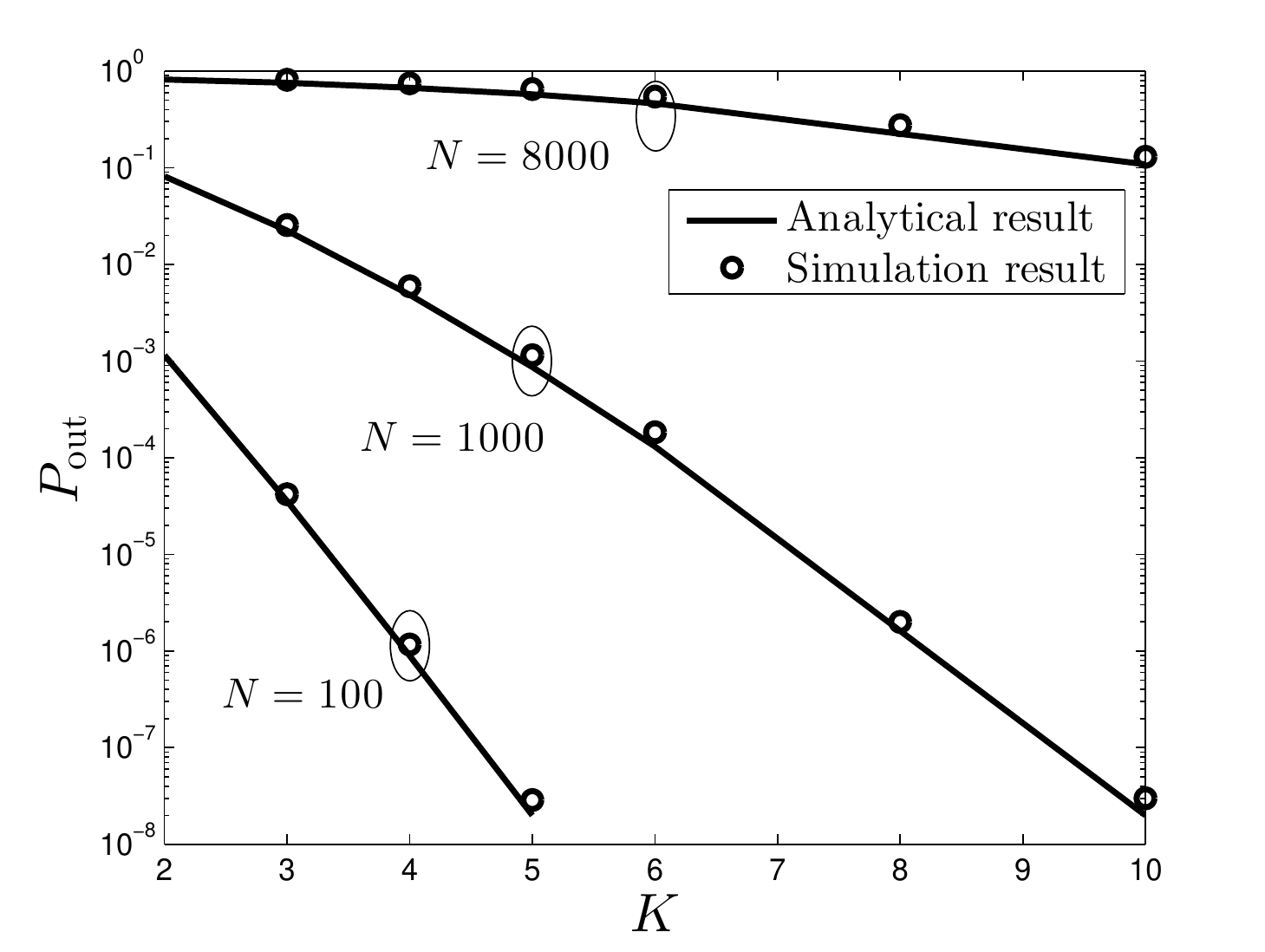}	
	\vspace*{-0.6cm}
	\caption{\small $\Pout$ versus $K$.}\label{fig:K_Pout}
	\endminipage
	\vspace*{-0.7cm}
\end{figure*}

\section{Conclusions}
In this paper, we have proposed a full-duplex BackCom network, where a novel TH-SS based multiple-access scheme is designed.
The scheme enables simultaneous forward/backward IT and can also suppress interference from coexisting links.
Moreover, the scheme not only supports dedicated ET for every symbol but also allows wireless energy harvesting from interference. 
Several interesting design insights are obtained, such as: 
the performance tradeoff between the backward IT and the forward IT/ET w.r.t. the reflection coefficients,
the performance tradeoff between the forward IT and ET w.r.t. the TH-SS sequence length for power constrained reader,
and also the performance tradeoff between the forward/backward IT and the forward ET w.r.t. the number of BackCom links.
{Although the proposed BackCom system enables simultaneously forward/backward IT and forward ET and suppresses inter-link interference, compared with the conventional BackCom system (i.e., the half-duplex system without spreading spectrum), it occupies a larger spectrum bandwidth due to the TH-SS scheme. We will investigate the spectrum efficiency of the proposed BackCom network for future study.}



\setcounter{equation}{0}
\renewcommand{\theequation}{A.\arabic{equation}}
{\small
\section*{Appendix A: Proof of Proposition~\ref{Tag_BER_fading}}

Based on \eqref{fading_tag_outage}, in order to calculate the BER for the forward IT, we need to calculate $\myprobability{\Prxc> \Prxa}$ and $\myprobability{\Prxc> \Prxd}$ as follows:
{\small 
\begin{equation} \label{P_tag_1_2}
\begin{aligned}
\myprobability{\Prxc> \Prxa}
&\stackrel{(a)}{=} 1- \myexpect{g_{21},q_2}{\int_{0}^{\infty}\exp\left(-\mu_{1} x\right) \mu_2 \exp\left(-\mu_{2} x\right) \mathrm{d}x}
= 1- \myexpect{g_{21},q_2}{\frac{\mu_2}{\mu_1+\mu_2}}\\
&= 1- \frac{1}{\rho} \left(\frac{d_{22}d_t}{d_{11}}\right)^\lambda
\exp\left(
\frac{d^\lambda_t}{\rho}
\left(
\left(\frac{d_{22}}{d_{21}}\right)^\lambda+
\left(\frac{d_{22}}{d_{11}}\right)^\lambda
\right)
\right)
\Gamma\left(
0,\ 
\frac{d^\lambda_t}{\rho}\left(\frac{d_{22}}{d_{21}}\right)^\lambda+
\left(\frac{d_{22}}{d_{11}}\right)^\lambda
\right),
\end{aligned}
\end{equation}
}
where $(a)$ is because given $g_{21}$ and $q_2$, both ${\vert f_{11} \vert^2}$
and $\left\vert {f_{21}} + \sqrt{\rho} {f_{22}} {g_{21}} q_2 \right\vert^2$ follow exponential distribution with parameters $\mu_1=1 \left/ \left(\frac{1}{d^\lambda_{11}}\right)\right.$ and $\mu_2=1\left/\left(\frac{1}{d^\lambda_{21}} + \frac{\rho \vert g_{21} \vert^2}{\left(d_{22} d_{t}\right)^{\lambda}}\right) \right.$, respectively.

Similarly, we have
\vspace{-0.5cm}
{\small
\begin{equation} \label{P_tag_2_2}
\begin{aligned}
&\myprobability{\Prxc> \Prxd}
&= 1\!-\!
\left(
\frac{d^\lambda_{22}}{d^\lambda_{12}\!+\!d^\lambda_{22}}
\!+\!\frac{d^\lambda_{t}}{\rho}  \frac{\frac{1}{d^\lambda_{11}d^\lambda_{22}}\!-\!\frac{1}{d^\lambda_{21}d^\lambda_{12}}}{\left(\frac{1}{d^\lambda_{12}}\!+\!\frac{1}{d^\lambda_{22}}\right)^2}
\exp\left(
\frac{d^\lambda_{t}}{\rho} \frac{\frac{1}{d^\lambda_{11}}\!+\!\frac{1}{d^\lambda_{21}}}{\frac{1}{d^\lambda_{12}}\!+\!\frac{1}{d^\lambda_{22}}}
\right)
\Gamma\left(
0,\ 
\frac{d^\lambda_{t}}{\rho} \frac{\frac{1}{d^\lambda_{11}}\!+\!\frac{1}{d^\lambda_{21}}}{\frac{1}{d^\lambda_{12}}\!+\!\frac{1}{d^\lambda_{22}}}
\right)
\right).
\end{aligned}
\end{equation}
}
Taking \eqref{P_tag_1_2} and \eqref{P_tag_2_2} into \eqref{fading_tag_outage}, the BER for the forward IT is obtained.

\setcounter{equation}{0}
\renewcommand{\theequation}{B.\arabic{equation}}
\section*{Appendix B: Proof of Proposition~\ref{fading_average_energy}}
Using the fact that 
$
\myexpect{}{\vert X + Y Z \vert^2} = \myexpect{}{\vert X \vert^2} +\myexpect{}{\vert Y \vert^2}\myexpect{}{\vert Z \vert^2}
$
if $X$, $Y$ and $Z$ are independent random variables and have zero mean, we thus have
$\myexpect{f_{11},q_2}{\Prxa} \!=\! \etabs \myP \frac{1}{d^\lambda_{11}}$,
$\myexpect{f_{11},f_{12},f_{21},f_{22},g_{21},q_2}{\Prxb} \!=\! \etabs \myP \left(\frac{1}{d^\lambda_{11}} \!+\!\frac{1}{d^\lambda_{21}}  \!+ \!\rho \left(\frac{1}{d^\lambda_{12}}\frac{1}{d^\lambda_{t}} \!+\! \frac{1}{d^\lambda_{22}} \frac{1}{d^\lambda_{t}}\right)  \right)$
$\myexpect{f_{21},f_{22},g_{21},q_2}{\Prxc} \!=\! \etabs \myP \left(\frac{1}{d^\lambda_{21}} \!+\! \rho \frac{1}{d^\lambda_{22}} \frac{1}{d^\lambda_{t}}\right),$
$\myexpect{f_{11},f_{12},g_{21},q_2}{\Prxd} \!=\! \etabs \myP \left(\frac{1}{d^\lambda_{11}} \!+\! \rho \frac{1}{d^\lambda_{12}} \frac{1}{d^\lambda_{t}}\right)$
$\myexpect{f_{21},f_{22},g_{21},q_2}{\Peh } = \Prxc/(1-\rho)$
Taking these results into \eqref{first_ave_energy}, the expected ETR is obtained.

\setcounter{equation}{0}
\renewcommand{\theequation}{C.\arabic{equation}}
\section*{Appendix C: Proof of Proposition~\ref{fading_outage}}
Given $g_{21}$ and $q_2$, then
$\Prxa/(\eta \myP)$ and $\Peh/(\eta \myP)$ are independent with each other and follow exponential distribution with mean parameters
$\left(1-\rho \right)\frac{1}{d^\lambda_{11}}$ and $ \left(\frac{1}{d^\lambda_{21}}+ \rho \frac{1}{d^\lambda_{22}} \frac{1}{d^\lambda_{t}} \vert g_{21} \vert^2\right)$, respectively.
Thus, we have
{\small
\begin{equation} \label{outage_prob_1}
\begin{aligned}
&\myprobability{ \frac{T}{N}  \left(\Prxa + \Peh\right) < \mathcal{E}_0 }
= \myexpect{g_{21},q_2}{
	\myprobability{ \left. \frac{\left(\Prxa + \Peh\right)}{\eta \myP} < \frac{N \mathcal{E}_0}{\eta \myP T} \right\vert g_{21},q_2}	
}\\
&= 1-  \mathbb{E}_{g_{21}}\left[
	\frac{\left(1-\rho \right)\frac{1}{d^\lambda_{11}} \exp\left(-\frac{\Xi}{\left(1-\rho \right)\frac{1}{d^\lambda_{11}}}\right)-
		\left(\frac{1}{d^\lambda_{21}}+ \rho \frac{1}{d^\lambda_{22}} \frac{1}{d^\lambda_{t}} \vert g_{21} \vert^2\right)
		\exp\left(-\frac{\Xi}{ \left(\frac{1}{d^\lambda_{21}}+ \rho \frac{1}{d^\lambda_{22}} \frac{1}{d^\lambda_{t}} \vert g_{21} \vert^2\right)}\right)
	}
	{\left(1-\rho \right)\frac{1}{d^\lambda_{11}}-  \left(\frac{1}{d^\lambda_{21}}+ \rho \frac{1}{d^\lambda_{22}} \frac{1}{d^\lambda_{t}} \vert g_{21} \vert^2\right)}
\right]\\
&=M\left(
\left(1-\rho \right)\frac{1}{d^\lambda_{11}},
0,
\frac{1}{d^\lambda_{21}},
\rho \frac{1}{d^\lambda_{22}} \frac{1}{d^\lambda_{t}} 
\right).
\end{aligned}
\end{equation}
}
where 
$
\Xi\triangleq \frac{N \mathcal{E}_0}{\eta \myP T}
$, and function $M(\cdot,\cdot,\cdot,\cdot)$ is defined in \eqref{outage_prob_2}.
Similarly, we can obtain $\myprobability{ \frac{T}{N}  \left(\Prxa + \Prxc\right) < \mathcal{E}_0 }$,
$\myprobability{ \frac{T}{N}  \left(\Prxd + \Peh\right) < \mathcal{E}_0 }$,
$\myprobability{ \frac{T}{N}  \left(\Prxd + \Prxc\right) < \mathcal{E}_0 }$,
$\myprobability{ \frac{T}{N} \Prxb \! < \!\mathcal{E}_0 }$.
Taking these results and \eqref{overlap_prob}, \eqref{outage_prob_1}, and \eqref{outage_prob_2} into \eqref{first_Pout}, the energy-outage probability is obtained.
}


\ifCLASSOPTIONcaptionsoff
\fi
\bibliographystyle{IEEEtran}

\end{document}